\documentclass[11pt,a4paper]{article}

\usepackage{authblk}
\usepackage{cite}                                       
\usepackage[latin1]{inputenc}
\usepackage[english]{babel}
\usepackage{amssymb,amsmath,amsbsy}
\usepackage{a4wide}
\usepackage[colorinlistoftodos]{todonotes}
\usepackage{enumitem}
\usepackage{color}
\usepackage{pdflscape}
\usepackage{tikz}
\usetikzlibrary{calc}
\newcommand{\tikzmark}[1]{\tikz[baseline,remember picture] \coordinate (#1) {};}

\newcommand{\K}{\mathcal K}
\newcommand{\Ph}{\mathcal P}

\newcommand{\R}{\mathcal R}

\newcommand{\C}{\mathcal C}

\renewcommand{\Re}{\mathfrak {Re}}
\renewcommand{\Im}{\mathfrak {Im}}

\newcommand{\tond}[1]{{\left(#1\right)}}
\newcommand{\quadr}[1]{{\left[#1\right]}}
\newcommand{\inter}[1]{{\langle#1\rangle}}
\newcommand{\graff}[1]{{\left\{#1\right\}}}

\newcommand{\derp}[2]{{\frac{\partial #1}{\partial #2}}}

\newcommand{\norm}[1]{\left\|#1\right\|}

\def\CC{\mathbb{C}}
\def\RR{\mathbb{R}}
\def\ZZ{\mathbb{Z}}

\def\Id{\mathbb I}
\def\eps{\epsilon}
\def\vphi{\varphi}
\def\im{i}

\newtheorem{theorem}{Theorem}[section]
\newtheorem{theorem*}{Theorem}
\newtheorem{lemma}{Lemma}[section]

\newtheorem{remark}{Remark}[section]

\newenvironment{proof}[0]{\noindent{\bf proof:}}{$\square$\par\medskip}

\title{On the nonexistence of degenerate phase-shift discrete
  solitons\\ in a dNLS nonlocal lattice.}

\author[1]{T. Penati}
\author[1]{M. Sansottera}
\author[1]{S. Paleari}
\author[2]{V. Koukouloyannis}
\author[3]{P.G. Kevrekidis.}

\affil[1]{Department of Mathematics ``F.Enriques'', Milano University, via Saldini 50, Milano, Italy, 20133}
\affil[2]{Department of Mathematics, Statistics and Physics, College of Arts and Sciences, Qatar University, P.O. Box 2713, Doha, Qatar}
\affil[3]{Department of Mathematics and Statistics, University of Massachusetts, Amherst, MA 01003-4515, USA}

\date{\Large{\tt Published version available here:
    https://doi.org/10.1016/j.physd.2017.12.012}}

\begin{document}

\maketitle

\begin{abstract}
 We consider a one-dimensional discrete nonlinear Schr{\"o}dinger
 (dNLS) model featuring interactions beyond nearest neighbors. We are
 interested in the existence (or nonexistence) of phase-shift discrete
 solitons, which correspond to four-sites vortex solutions in the
 standard two-dimensional dNLS model (square lattice), of which this
 is a simpler variant.  Due to the specific choice of lengths of the
 inter-site interactions, the vortex configurations considered present
 a degeneracy which causes the standard continuation techniques to be
 non-applicable.

In the present one-dimensional case, the existence of a conserved
quantity for the soliton profile (the so-called density current),
together with a perturbative construction, leads to the nonexistence
of any phase-shift discrete soliton which is at least $C^2$ with
respect to the small coupling $\eps$, in the limit of vanishing
$\eps$. If we assume the solution to be only $C^0$ in the same limit
of $\eps$, nonexistence is instead proved by studying the bifurcation
equation of a Lyapunov-Schmidt reduction, expanded to suitably high
orders. Specifically, we produce a nonexistence criterion whose
efficiency we reveal in the cases of partial and full degeneracy of
approximate solutions obtained via a leading order expansion.
\end{abstract}

\bigskip

KEYWORDS: discrete Non-Linear Schr{\"o}dinger, discrete solitons,
discrete vortex, current conservation, perturbation theory,
Lyapunov-Schmidt decomposition.


\section{Introduction}
\label{s:0}

The discrete nonlinear Schr{\"o}dinger (DNLS) is a prototypical
nonlinear lattice dynamical model whose analytical and numerical
tractability has enabled a considerable amount of progress towards
understanding lattice solitary waves/coherent
structures~\cite{Kev_book09}.  Its apparent simplicity in
incorporating the interplay of nonlinearity and discrete dispersion,
together with its relevance as an approximation of optical waveguide
systems~\cite{chris03, sukh03, led08} and atomic systems in optical
lattices~\cite{mor06} have significantly contributed to the popularity
of the model. Moreover, its ability to capture numerous linear and
nonlinear experimentally observed features has made it a useful
playground for a diverse host of phenomena. Such examples include, but
are not limited to discrete diffraction~\cite{yaron} and its
management~\cite{yaron1}, discrete solitons~\cite{yaron,yaron2,yaron3}
and vortices~\cite{neshev,fleischer}, Talbot revivals~\cite{christo2},
and $\mathcal{PT}$-symmetry and its breaking~\cite{kip}, among many
others.

Among the solutions that are of particular interest within the dNLS
model are the so-called phase-shift ones, for which the solution does
not bear a simply real profile (with phases $0$ and $\pi$, or sites
that are in- and out-of phase), but rather is genuinely complex
featuring nontrivial phases~\cite{Kev_book09}. Such solutions are more
widely known as ``discrete vortices'' due to the fact that in order to
ensure single-valuedness of the solution, upon rotation over a closed
discrete contour, they involve variation of the phase by a multiple of
$2 \pi$. Such waveforms were originally proposed theoretically
in~\cite{vort1,vort2} and subsequently observed experimentally,
especially in the setting of optically induced photorefractive
crystals~\cite{vort3,vort4}. A systematic analysis of their potential
existence in the standard nearest-neighbor dNLS model was provided
in~\cite{PelKF05b} and the relevant results, not only for 2d square
lattices, but also for lattices of different types (triangular,
honeycomb, etc.)  were subsequently summarized
in~\cite{Kev_book09}. One of the main findings in this context are
that the (symmetric) square vortices (with $\pi/2$ phase
shifts between adjacent excited nodes) indeed persist at high order in
case of the square lattice for different topological charges. Another
relevant conclusion was that hexagonal and honeycomb lattices present
the potential for vortices of topological charge both $S=1$ and $S=2$.
Intriguingly, among the two and for focusing nonlinearity, the latter
was found to be more stable than the former (whereas the stability
conclusions were reversed in the case of self-defocusing
nonlinearity.

However, the relevant analysis poses some intriguing mathematical
questions. In particular, the consideration of the most canonical
$4$-site vortex in the 2d nearest-neighbor dNLS reveals (due to the
relative phase change between adjacent sites of $\pi/2$) a degeneracy
of the relevant waveforms and of their potential persistence. By
degeneracy, here we mean that the standard approach of looking for
critical points of the averaged (along the unperturbed periodic
solution) perturbation (see \cite{Kap01,Kev09}), produces
one-parameter families of solutions, where the Implicit Function
Theorem cannot be applied. The tangent direction to the family
represents a direction of degeneracy.  Typically this calculation
proceeds from the so-called anti-continuum limit~\cite{Mak96} in
powers of the coupling. Given then the degeneracy of these vortex
states, a lingering question is whether such states will persist {\it
  to all orders} or whether they may be destroyed (i.e., the relevant
persistence conditions will not be satisfied) at a sufficiently high
order. While to all the leading orders considered
in~\cite{PelKF05b,Kev_book09} these solutions persist,
in case of high degeneracy expansions
up to high orders ($\eps^6$ in the $\pi/2$-vortex, $\eps$ being the
perturbation parameter) are necessary to reach a definitive answer to
this question.

Inspired by the inherent difficulty of tackling the 2d problem, here
we will opt to examine a simpler 1d problem. As a ``caricature'' of
the 2d interaction, where the fourth site of a given square contour
couples {\it back} into the first site, we choose to examine a {\it
  one-dimensional} model involving interactions not only of nearest
neighbors (NN), but also of neighbors that are next-to-next-nearest
(NNNN) ones. In principle, isolating a quadruplet of sites, we
reconstruct a geometry similar to the 2d contour. The question that we
ask in this simpler (per its 1d nature) setting is whether phase-shift
solutions will exist. Surprisingly, and differently from the 2d
scenario where at least the $\pi/2$ vortex was shown to exist, the
answer that we find here is always in the negative. Using both a more
direct, yet more restrictive, method involving a conserved flux
quantity, as well as a more elaborate, yet less restrictive technique
based on Lyapunov-Schmidt reductions, we illustrate that such vortical
states are always precluded from existence at a sufficiently high
order. Since the continuation problem requires to explore the
persistence of the solution at high orders in $\eps$, where the
differences with respect to the 2d model, in terms of lattice shape
and interaction among sites, play a role, it is perhaps not surprising
that we report here a different result in comparison to the
$\pi/2$-vortex solution of the 2d lattice.

Our presentation will be structured as follows. In section 2, we will
discuss the model and the principal result. In section 3, we will
provide a perturbative approach which, combined with the conservation
of the density current, leads towards the formulation of a finite
regularity ($\C^2$) version of this result. In section 4, we will
overcome the technical limitation of the above regularity requirement by extending
considerations to Lyapunov-Schmidt decompositions. Finally, in section
5, we will summarize our findings and present our conclusions, as well
as some emerging questions for future work.

\section{Theoretical Setup and Principal Results}
\label{s:1}

As explained in the previous section, the aim of the work is to
investigate the existence of discrete solitons in the NN and NNNN dNLS
model of the form:
\begin{equation}
  \label{e.dNLS.eqs}
  \im\dot\psi_j = \psi_j -\frac\eps2\quadr{(\Delta_1+\Delta_3)\psi}_j +
  \frac34\psi_j|\psi_j|^2 \ ,
\qquad
  (\Delta_l\psi)_j := \psi_{j+l}-2\psi_j+\psi_{j-l}\ ,
\end{equation}
with vanishing boundary conditions at infinity $\psi\in\ell^2(\CC)$,
in the anti-continuum limit, namely $\eps\to 0$.  The equations can be
written in the Hamiltonian form $\im\dot\psi_j =
\derp{\K}{\overline\psi_j}$ with
\begin{equation}
  \label{e.KdNLS}
  \K = \sum_{j\in\ZZ}|\psi_j|^2 + 
  \frac{\eps}2\sum_{j\in\ZZ}\tond{|\psi_{j+1}-\psi_j|^2+|\psi_{j+3}-\psi_j|^2}
  + \frac38\sum_{j\in\ZZ}|\psi_j|^4\ .
\end{equation}

The original motivation leading to the study of the above model was
the continuation of 4-site multibreathers in the Klein-Gordon model
with NN and NNNN linear
interactions of equal strength (once again motivated by
the 2d problem), namely
\begin{equation}
  \label{e.H}
  H = \frac12\sum_{j\in\ZZ} \tond{ y^2_j + x^2_j } +
    \frac{\eps}2\sum_{j\in\ZZ} \tond{ (x_{j+1}-x_j)^2 + (x_{j+3}-x_j)^2 } +
  \frac{1}4\sum_{j\in\ZZ} x_j^4 \ .
\end{equation}
On one hand, this Hamiltonian represents a special case of the
three-parameters family of up to next-to-next-nearest neighbor
Klein-Gordon Hamiltonians
\begin{equation*}
  \label{e.H.kappa12}
  H_{\kappa_1,\kappa_2,\kappa_3} = \frac12\sum_{j\in\ZZ} \tond{ y^2_j + x^2_j } +
    \frac{\eps}2\sum_{j\in\ZZ} \sum_{m=1}^3  \kappa_m (x_{j+m}-x_j)^2 +
  \frac{1}4\sum_{j\in\ZZ} x_j^4 \ ,
\end{equation*}
where the value of $\kappa_m$ controls the range $m$ interaction
strength. The emergence of phase-shift multibreathers for values of
$\kappa_m$ large enough has been recently investigated
\cite{KouKCR13,Rap13}, while, the nonexistence of phase-shift
multibreathers in the standard Klein-Gordon model with only
nearest-neighbours interactions has been proved in \cite{Kou13}.

On the other hand, the particular choice in \eqref{e.H},
i.e. $\kappa_2=0,\;\kappa_1=\kappa_3=1$, comes from the idea of
reproducing in the 1D setting those vortex-like 4-sites interactions
which are peculiar of a squared-lattice 2D model, with a standard
nearest-neighbours interaction.  In the latter case the approximate
4-sites vortex solutions (obtained with $\epsilon=0$) are degenerate
objects~\cite{CueKKA11}, as a standard averaging analysis shows (see
\cite{Mak96,Aub97,Ahn98,AhnMS02}).  Thus it is not possible to apply
the usual implicit function theorem for their continuation.
Nevertheless, it is well known (see \cite{PalP14,PelPP16}) that for
sufficiently low energies ($E\ll 1$) and in a suitable anti-continuum
limit regime (namely $\eps\ll E$), \eqref{e.H} can be approximated
by~\eqref{e.KdNLS}. Indeed, the dNLS Hamiltonian \eqref{e.KdNLS} turns
out to be a resonant normal form of \eqref{e.H}; this is evident by
averaging both the coupling term and the nonlinear term with respect
to the periodic flow given by the harmonic part of the Hamiltonian
\eqref{e.H}. The canonical change of coordinates which averages the KG
model \eqref{e.H}, generates both the normal form Hamiltonian
\eqref{e.KdNLS} and also some remainder terms: the energy and small
coupling regimes we are mentioning are those which guarantee the
remainder terms to be much smaller, hence negligible, than the normal
form \eqref{e.KdNLS}. The interesting point in this normal form
perspective is that the degeneracy is present also in the
corresponding 4-sites discrete vortices of the normal form; the
latter, being a model of the dNLS type, allows for an easier, thus
more accurate and complete, analysis via perturbation theory. We plan
to exploit the present results in order to transfer such an analysis
to the original model in a forthcoming paper \cite{PKSPK17kg}.

Finally, besides our original motivation of establishing a connection
between solutions of the Klein-Gordon model with those of the
corresponding dNLS normal form, we stress that the study of discrete
solitons in beyond nearest neighbor 1-dimensional dNLS models has
received some attention in the recent
literature~\cite{ChoCMK11,Kev09,BenCMP15}. We expect that the results
we are going to present can provide some additional insight on this
and similar topics, as, e.g., for discrete solitons in zigzag dNLS
models~\cite{EfrC02}.

In order to state the results of the present paper, we recall that
we are interested in periodic solutions of \eqref{e.dNLS.eqs} of the
form
\begin{equation}
\label{e.ansatz}
\psi_j = e^{-\im \lambda t}\phi_j\ ,\qquad
\graff{\phi_j}_{j\in\ZZ}\in\ell^2(\CC)\ ;
\end{equation}
by inserting the previous ansatz in \eqref{e.dNLS.eqs} one gets the
{\sl stationary equation}
\begin{equation}
  \label{e.dNLS.phi}
  \omega\phi_j =
    -\frac\eps2\Bigl[ (\Delta_1+\Delta_3)\phi \Bigr]_j +
    \frac34\phi_j|\phi_j|^2\ ,
\qquad
  \text{where}
\ \omega := \lambda-1\ .
\end{equation}

Specifically, being interested in many-site discrete solitons and in
particular in vortex-like 4-site solutions, among the infinitely many
trivial solutions of the unperturbed case, i.e., \eqref{e.dNLS.phi}
with $\eps=0$, we investigate a 4-dimensional torus (bearing, once
again, in mind the analogy with the 2d lattice). Furthermore, since
the discrete soliton solutions we are considering are single frequency
solutions, namely in the form of standing waves as in
\eqref{e.ansatz}, we need to set in the anti-continuum limit a fixed
common amplitude $R$ (which is fixed by the frequency $\lambda$, see
\eqref{e.om_lam}), so that the unperturbed solutions read

\begin{equation}
  \label{e.4Dtorus}
  \phi_j^{(0)} = \begin{cases}
    R e^{\im \theta_j} \, , &j\in S\ ,       \\
    0               \, , &j\not\in S\ ,
  \end{cases}
\qquad
  \text{where}\ S=\graff{1,2,3,4} \ \text{and}\ R>0\ .
\end{equation}
All these orbits are uniquely defined except for a phase
shift, due to the action of the symmetry $e^{\im\sigma}$ along the
orbit, which corresponds to a change of the initial configuration in
the ansatz \eqref{e.ansatz}.

We are interested in the investigation of the \emph{breaking of such a
  completely resonant lower dimensional torus}, i.e. we want to
determine which solutions are going to survive as $\eps\not=0$, at
fixed $\omega$ (and hence at fixed period). Before continue any
further, we introduce the phase shifts between successive sites as
\begin{equation}
\label{e.phidef}
\varphi_j := \theta_{j+1}-\theta_j\  .
\end{equation}
Our principal finding can be encapsulated then in the following
statement, which provides a negative answer to the possible existence
of phase-shift discrete solitons:
\begin{theorem}
\label{t.nonexistence}
For $\eps$ small enough ($\eps\not=0$), the only unperturbed solutions
\eqref{e.4Dtorus} that can be continued at fixed period to solutions
$\phi(\eps)$ of \eqref{e.dNLS.phi}, correspond to
$\varphi_j\in\graff{0,\pi}$ and $j\in S'=\{1,2,3\}$.
\end{theorem}

We stress that the existence of the trivial phase-shift solutions
$\varphi_j\in\graff{0,\pi}$ with $j\in S'$ can be obtained with
standard arguments, restricting to real\footnote{Briefly, the
  idea is that restricting to real solutions of \eqref{e.dNLS.phi},
  the kernel directions (which are purely imaginary) are removed and
  the implicit function theorem can be applied. From a variational perspective, the
  restriction provides a critical point on the real phase space
  $\ell^2(\RR)$ which is indeed a critical point also on the complex
  phase space $\ell^2(\CC)$, due to the invariance of the Hamiltonian
  under conjugacy (see \cite{BamPP10}, Section 5).} solutions $\phi_j$
of the stationary equation \eqref{e.dNLS.phi} (see \cite{PelKF05}
Proposition 2.1 and references therein).


In order to prove a nonexistence result like the one stated above, a
natural strategy could be the use a perturbative approach, i.e. 
expand both the candidate solution and the stationary equation in
powers of $\eps$ and look for an obstruction to the solution
order-by-order. The first issue in such a procedure is that the
obstruction may appear at a high order; indeed that's the case in
our model. We overcome this problem by exploiting the existence
of a conserved quantity, the so-called density current (see
  \cite{HenT99,Kev_book09}),
\begin{equation}
  \label{e.current}
  J_j = \Im\Bigl(  
    \phi_{j-1}\overline\phi_j + \phi_{j-3}\overline\phi_j +
    \phi_{j-2}\overline\phi_{j+1} + \phi_{j-1}\overline\phi_{j+2}
    \Bigr)\ ,\quad\text{for}\ j\in\ZZ\ ,
\end{equation}
that allows us to explore the implications of the obstruction (towards
this conservation law), greatly simplifying the calculations and
reveal them in a much lower order.  A detailed description is
presented in the first part of the paper (Section~\ref{s:break}, see
Theorem~\ref{t.notsmooth}).  We remark that this idea is expected to
be applicable more broadly to several different one-dimensional
models; clearly, due to the lack of a straightforward analogue of the
density current~\eqref{e.current} in two or three dimensions, there is
a natural difficulty towards plainly extending this approach to higher
(e.g., two or three) dimensional models.  We nevertheless believe that
the present analysis is useful to show how degenerate objects, which
are likely to exist at lower order expansions, might have some
obstruction to their continuation at higher order in perturbation
theory.

Besides the technical difficulty, there is an intrinsic problem in the
perturbative approach: it cannot provide a {\sl complete} non existence
result.  Indeed one has to assume enough regularity of the solution to
perform its expansion; in our case, due to the use of the density
current, obstructions show up at the second order, thus
Theorem~\ref{t.notsmooth} states that there cannot be any $\C^2$
discrete soliton solution in the perturbation parameter $\eps$,
except for the standard in-phase or out-of-phase ones.
The complete proof of Theorem~\ref{t.nonexistence}
requires to develop a Lyapunov-Schmidt decomposition which allows to
rely on the regularity of the equations instead of the regularity of
the solution, as shown in the second part of the paper (Section
\ref{s:LS}).  An interesting remark is that the unsolvable systems of
either linear or quadratic equations which appear in the density
current expansion of Section~\ref{s:break}, are essentially included
at the level of either the linearized (as in \cite{PelKF05b}) or
quadratic bifurcation\footnote{By bifurcation equation we mean, as
  usual in this approach, the kernel equation of the Lyapunov-Schmidt
  decomposition. We stress here that the linearized bifurcation
  equation we get is the same as obtained in \cite{PelKF05b}. Instead of
  working with the variational formulation of the problem, we here
  prefer to perform the Taylor expansion at the level of the
  stationary equation, since we are not interested in the linear
  stability but only in the continuation of all the candidate orbits.}
equation. Hence, {\sl a posteriori}, we could say that the
obstructions to the continuation which arise in the density current
expansion, {\sl represent an effective non-existence criterion}.


\section{A perturbative approach: finite regularity result}
\label{s:break}

In the present Section we tackle the problem with a perturbative
approach.  For this purpose, we assume to deal with a continuation of
$\{\phi_j(\eps)\}_{j\in\ZZ}$ which is at least $\C^2$ in $\eps$ and
write 
\begin{equation}
\label{e.exp.phi}
\phi_j = \phi^{(0)}_j + \eps\phi^{(1)}_j + \eps^2\phi^{(2)}_j +
o(\eps^2)\ .
\end{equation}
The continuation is assumed to be performed at fixed period
(frequency). The results (that are weaker than
Theorem~\ref{t.nonexistence}) are collected in the following

\begin{theorem}
\label{t.notsmooth}
For $\eps$ small enough ($\eps\not=0$), the only unperturbed solutions
\eqref{e.4Dtorus} that can be continued at fixed period to $\C^2$
solutions $\phi(\eps)$ of \eqref{e.dNLS.phi} 
correspond to $\varphi_j\in\graff{0,\pi}, j\in S'$.
\end{theorem}

As anticipated in the Introduction, a key point is the fact
that~\eqref{e.dNLS.phi} preserves the density current, precisely
\begin{lemma}Let $\{\phi_j(\eps)\}_{j\in\ZZ}$ solve \eqref{e.dNLS.phi}, then
\begin{equation*}
  J_j := \Im\Bigl(  
    \phi_{j-1}\overline\phi_j + \phi_{j-3}\overline\phi_j +
    \phi_{j-2}\overline\phi_{j+1} + \phi_{j-1}\overline\phi_{j+2}
    \Bigr)
    \equiv 0\ ,\quad \forall j\in\ZZ\ .
\end{equation*}
\end{lemma}

\begin{proof}
Let us define
\begin{displaymath}
a_n:=\Im(\phi_{n-1}\overline \phi_n)\quad\text{and}\quad
b_n:=\Im(\phi_{n-3}\overline \phi_n)\ ;
\end{displaymath}
then it is easy to see, e.g. by multiplying \eqref{e.dNLS.phi} by
$\overline\phi_j$ and exploiting the reality of some of the terms
thus obtained, that
\begin{equation}
\label{e.A.1}
a_{n}+b_{n} = a_{n+1}+b_{n+3}\ .
\end{equation}
If we define the density current as $J_n := a_{n}+b_{n}
+b_{n+1}+b_{n+2}$, by adding the quantity $b_{n+1} + b_{n+2}$ to both
the l.h.t. and r.h.t. of \eqref{e.A.1}, we get $J_n = J_{n+1}$.  The
hypothesis $\graff{\phi_j}_{j\in\ZZ}\in\ell^2(\CC)$, imply
$J_j = 0$, $\forall j\in\ZZ$.
\end{proof}

We present the detailed proof of Theorem~\ref{t.notsmooth} in the rest
of this Section. First, starting from the zero order expansion, we
determine the candidate solutions.  Then we use the expansion of the
stationary equation~\eqref{e.dNLS.phi}, together with the conserved
quantity~\eqref{e.current}, in order to find an incompatibility
condition and exclude all the solutions prohibited by
Theorem~\ref{t.notsmooth}. The advantage and novelty of this
approach lies in the fact that such an incompatibility is
revealed at a considerably lower order in the
conserved quantity, in comparison to the original equations
of motion.


\subsection{Zero order expansion and candidate solutions}
\label{s:zero.order}

The stationary equation \eqref{e.dNLS.phi} at order zero gives the
uncoupled system
\begin{equation}
\label{e.order.0}
\omega\phi_j^{(0)} = \frac34\phi_j^{(0)}\bigl|{\phi_j^{(0)}}\bigr|^2\ ,
\end{equation}
which is trivialy invariant under the action of $e^{\im\sigma}$. By
using \eqref{e.4Dtorus}, it provides the frequency $\lambda$ of the
orbit, and its detuning $\omega$ from the linear frequency 1, namely
\begin{equation}
\label{e.om_lam}
\omega = \frac34 R^2\quad\text{and}\quad \lambda =1+\frac34 R^2\ .
\end{equation}
The conservation law \eqref{e.current} at order zero gives
\begin{equation}
  \label{e.current.o0}
  J_j^{(0)} := \Im\Bigl(
    \phi_{j-1}^{(0)}\overline\phi_j^{(0)} +
    \phi_{j-3}^{(0)}\overline\phi_j^{(0)} +
    \phi_{j-2}^{(0)}\overline\phi_{j+1}^{(0)} +
    \phi_{j-1}^{(0)}\overline\phi_{j+2}^{(0)}
  \Bigr) = 0 \ .
\end{equation}
Recalling the form of the ansatz solution \eqref{e.4Dtorus},
eq.~\eqref{e.current.o0} is identically satisfied for $j\not\in S$.
Instead, for $j\in S$ we get
  \begin{equation}
    \label{e.shifts}
\sin\tond{\varphi_1} = \sin\tond{\varphi_2} = \sin\tond{\varphi_3} = 
- \sin\tond{\varphi_1 + \varphi_2 + \varphi_3}\ .
\end{equation}

\begin{remark}
The above systems of equations for the phase shifts can be obtained
with a different procedure, namely the leading order approximation of
a variational argument (see \cite{Kap01,Kev09}), which applies to
system with symmetries. The solutions can indeed be obtained as
critical points of the perturbed energy
\begin{displaymath}
H_1(\psi) =
\eps\sum_j\quadr{|\psi_{j+1}-\psi_j|^2+|\psi_{j+3}-\psi_j|^2} \ ,
\end{displaymath}
restricted to the unperturbed solution $\phi^{(0)}e^{-\im\lambda t}$
(or averaged over the unperturbed periodic orbit
$\phi^{(0)}e^{-\im\lambda t}$), giving in our case
\begin{equation}
\label{e.F}
F(\varphi_1,\varphi_2,\varphi_3) := H_1(\phi^{(0)}e^{-\im\lambda t}) =
-2R^2\tond{\cos\tond{\varphi_1} + \cos\tond{\varphi_2} +
  \cos\tond{\varphi_3} + \cos\tond{\varphi_1+\varphi_2+\varphi_3}}\ ,
\end{equation}
apart from useless constant terms depending on $R$. However, in our
perspective, it is more important to stress that the system
\eqref{e.shifts} also represents the first term of the bifurcation
equation (in the Lyapunov-Schmidt decomposition) expanded both in
$\eps$ and in the kernel variables (see \eqref{e.bif.point}).
\end{remark}

The first of~\eqref{e.shifts} provides four, one-parameter, families
of solutions
$(\varphi_1,\varphi_2(\varphi_1),\varphi_3(\varphi_1))$. By replacing
$\varphi_1$ by $\varphi$ we get precisely
\begin{equation}
(\varphi,\varphi,\varphi)\ ,\quad (\varphi,\varphi,\pi-\varphi)\ ,\quad
  (\varphi,\pi-\varphi,\varphi)\ ,\quad (\varphi,\pi-\varphi,\pi-\varphi)\ . 
\label{families}
\end{equation}

By plugging the first one of (\ref{families}), i.e., $(\varphi,\varphi,\varphi)$ in the
second equation of \eqref{e.shifts} we get $\sin\tond{\varphi} =
-\sin\tond{3\varphi}$, thus we obtain
\begin{displaymath}
(\varphi_1,\varphi_2,\varphi_3) \in \left\{
  (0,0,0)\ , (\pi,\pi,\pi)\ , \pm\tond{\frac\pi2,\frac\pi2,\frac\pi2}
\right\}\ .
\end{displaymath}
The first two solutions represent the discrete in-phase and the
alternating-phase solitons, while the other two are the
discrete vortex solitons (the phase shift solutions).
In the latter cases, i.e. the solution we
call symmetric vortices, we will conveniently use the following choice
of phases:
\begin{equation}
  \label{e.vortex}
  \left( \theta_1,\theta_2,\theta_3,\theta_4 \right) =
  \pm \left( 0,\frac\pi2,\pi,\frac{3\pi}2 \right)\ .
\end{equation}

The remaining three families, when plugged in the second of
\eqref{e.shifts} give the identity $\sin\tond{\varphi_1} =
\sin\tond{\varphi_1}$, hence we get the following three 1-parameter
families of solutions which are referred to as asymmetric vortices
(see~\cite{PelKF05b} for a similar example in the nearest neighbor case)
\begin{equation}
  \label{e.sol-shifts}
  F_1: (\varphi,\pi-\varphi,\pi-\varphi) \ ,\quad
  F_2: (\varphi,\varphi,\pi-\varphi)     \ ,\quad
  F_3: (\varphi,\pi-\varphi,\varphi)     \ .
\end{equation}
Remarkably, the three families intersect in the two previously
obtained symmetric vortex solutions $ \pm\tond{\frac\pi2,\frac\pi2,\frac\pi2}$,
while the in-phase/alternating phase solutions do not belong to any of
these families, thus being isolated solutions. Moreover, the above
families carry also all the other 3 couples of solutions with phase
shifts $\varphi_j\in\graff{0,\pi}$.

The above considerations provide us with the complete list of the
phase shift, discrete vortex solutions that are admissible for the
continuation with respect to $\eps$. On a three-dimensional torus
$\mathbb{T}^3$, representing all the possible phase differences (which
means, the original $\mathbb{T}^4$ with a quotient with respect to the
gauge symmetry), we have 2 isolated solutions, and three 1-parameter
families intersecting in the 2 symmetric vortex solutions. Thus, the 2
vortex solutions will be fully degenerate, while any other asymmetric
vortex solution on a family will be partially degenerate. This can be
seen explicitly by using the approach developed in
\cite{Kap01}. Indeed, it turns out that $(0,0,0)$ and $(\pi,\pi,\pi)$
are non degenerate (absolute) extrema -- respectively a maximum and a
minimum -- of $F$ (see \eqref{e.F} above) on the torus $\mathbb{T}^3$,
while the mixed standard solution are partially degenerate relative
extrema. The two symmetric vortex solutions are fully degenerate, since the
Hessian $\tond{D^2F}\tond{\pm \tond{\frac\pi2,\frac\pi2,\frac\pi2}}$
is the null matrix.  This scenario in the 2d, nearest-neighbor dNLS
was dubbed the super-symmetric case~\cite{PelKF05b}.  In all the other
asymmetric vortex solutions, the Hessian $\tond{D^2F}$ has only 1 zero
eigenvalue which corresponds to the tangent direction to the family
one considers Notice that the gauge symmetry direction is absent from
$F$, given its dependence on the variables $\varphi$.  This fact can
also be verified by a direct calculation.

\begin{remark}
The above analysis bears extensive similarities with the study of the vortex
solutions on a squared-lattice 2D dNLS model, performed in
\cite{PelKF05b}. The above classifications of 1-parameter families
\eqref{e.sol-shifts} represents indeed the families (3.8)-(3.10) in
that work.
\end{remark}


\subsection{First and second order expansions}

The strategy we are going to follow is based on the expansion in
$\epsilon$ up to second order of both the stationary equation
\eqref{e.dNLS.phi} and the density current \eqref{e.current}. Hence we
start providing here the explicit expansions that will be used in the
following parts.

\subsubsection{First order equation's structure:}

The equation \eqref{e.dNLS.phi} at first order reads
\begin{equation*}
\omega\phi_j^{(1)} = -\frac12\tond{\phi^{(0)}_{j+3} +
  \phi^{(0)}_{j+1} + \phi^{(0)}_{j-1} + \phi^{(0)}_{j-3}} +
2\phi^{(0)}_j + \frac34\tond{2\phi_j^{(1)}\bigl|\phi_j^{(0)}\bigr|^2 +
  \bigl({\phi_j^{(0)}}\bigr)^2\overline{\phi}_j^{(1)}}\ ;
\end{equation*}
which, separating between the core sites of the discrete soliton and the external
ones, becomes
\begin{align}
  -\omega\tond{\phi_j^{(1)} + e^{2\im\theta_j}\overline{\phi}_j^{(1)}} &=
  -\frac12\tond{\phi^{(0)}_{j+3} + \phi^{(0)}_{j+1} + \phi^{(0)}_{j-1}
  + \phi^{(0)}_{j-3}} + 2\phi^{(0)}_j \ ,
  & j\in S\ , \label{e.order.1.S}
\\
  \omega\phi_j^{(1)} &= -\frac12\tond{\phi^{(0)}_{j+3} +
  \phi^{(0)}_{j+1} + \phi^{(0)}_{j-1} +
  \phi^{(0)}_{j-3}} \ ,
 & j\not\in S\ . \label{e.order.1.notS}
\end{align}
Equation \eqref{e.order.1.S} is a linear equation in $\phi_j^{(1)}$
with a 1-dimensional kernel for any $j\in S$ given by
\begin{equation}
\label{e.Kern.j}
\phi_{j,Ker}^{(1)} = \im e^{\im \theta_j} \ ,
\end{equation}
and following the ideas of \cite{PelKF05b}, part 2, we can find a
particular solution of \eqref{e.order.1.S} in the form
\begin{equation}
  \label{e.order.1.S.sol}
  \phi_j^{(1)} = e^{\im\theta_j} u_j^{(1)} \ ,
  \qquad
  u_j^{(1)}\in\RR \ .
\end{equation}

\subsubsection{Second order equation's structure:}

Equation \eqref{e.dNLS.phi} at second order, directly split between core
and external sites, reads
\begin{align}
  -\omega\tond{\phi_j^{(2)} + e^{2\im\theta_j}\overline{\phi}_j^{(2)}} =&
  -\frac12\tond{
    \phi^{(1)}_{j+3} + \phi^{(1)}_{j+1} + \phi^{(1)}_{j-1} + \phi^{(1)}_{j-3}}
  + 2\phi_j^{(1)} +  \nonumber
\\
  &\quad-\frac34\tond{
    2\phi_j^{(0)}\bigl|\phi_j^{(1)}\bigr|^2 + \bigl({\phi_j^{(1)}}\bigr)^2\overline{\phi}_j^{(0)}
  } \ ,
  & j\in S\ , \label{e.order.2.S}
\\
  \omega\phi_j^{(2)} =& -\frac12\tond{
    \phi^{(1)}_{j+3} + \phi^{(1)}_{j+1} + \phi^{(1)}_{j-1} + \phi^{(1)}_{j-3}} +
  2\phi_j^{(1)} \ , 
 & j\not\in S\ . \label{e.order.2.notS}
\end{align}
The kernel part associated to \eqref{e.order.2.S} being the same as in
\eqref{e.order.1.S}, gives
\begin{equation}
  \label{e.order.2.S.sol}
  \phi_j^{(2)} = e^{\im\theta_j} u_j^{(2)} \ ,
\qquad
  u_j^{(2)}\in\RR \ .
\end{equation}

\subsubsection{Density current expansion:}

The conservation law \eqref{e.current} at first and second orders, respectively
reads
\begin{equation}
  \label{e.current.o1}
  \begin{aligned}
  J_j^{(1)} := \Im\Bigl( &
    \phi_{j-1}^{(0)}\overline\phi_j^{(1)} +
    \phi_{j-3}^{(0)}\overline\phi_j^{(1)} +
    \phi_{j-2}^{(0)}\overline\phi_{j+1}^{(1)} +
    \phi_{j-1}^{(0)}\overline\phi_{j+2}^{(1)} + \Bigr.
\\
  &\quad+  \Bigl.
    \phi_{j-1}^{(1)}\overline\phi_j^{(0)} +
    \phi_{j-3}^{(1)}\overline\phi_j^{(0)} +
    \phi_{j-2}^{(1)}\overline\phi_{j+1}^{(0)} +
    \phi_{j-1}^{(1)}\overline\phi_{j+2}^{(0)} \Bigr) = 0 \ ,
  \end{aligned}
\end{equation}

\begin{equation}
  \label{e.current.o2}
  \begin{aligned}
    J_j^{(2)} := \Im\Bigl( &
      \phi_{j-1}^{(0)}\overline\phi_j^{(2)} +
      \phi_{j-1}^{(1)}\overline\phi_j^{(1)} +
      \phi_{j-1}^{(2)}\overline\phi_j^{(0)} +
      \phi_{j-3}^{(0)}\overline\phi_j^{(2)} +
      \phi_{j-3}^{(1)}\overline\phi_j^{(1)} +
      \phi_{j-3}^{(2)}\overline\phi_j^{(0)} + \Bigr.
\\
  & \quad+\Bigl.
      \phi_{j-1}^{(0)}\overline\phi_{j+2}^{(2)} +
      \phi_{j-1}^{(1)}\overline\phi_{j+2}^{(1)} +
      \phi_{j-1}^{(2)}\overline\phi_{j+2}^{(0)} +
      \phi_{j-2}^{(0)}\overline\phi_{j+1}^{(2)} +
      \phi_{j-2}^{(1)}\overline\phi_{j+2}^{(1)} +
      \phi_{j-2}^{(2)}\overline\phi_{j+2}^{(0)} \Bigr) = 0\ .
  \end{aligned}
\end{equation}


\subsection{The $(0,0,0)$ solution}

In what follows, we start showing how our approach is
implemented in the easiest case of
a non degenerate discrete soliton, corresponding to
$\theta_1=\theta_2=\theta_3=\theta_4$ where the compatibility
of the solution expansion with the density current expansion persists.  This is
a case where the implicit function theorem can be applied
to continue the unperturbed solution, even without exploiting the
restriction to real solutions of \eqref{e.dNLS.phi}. However,
we think it is instructive to start to
apply here our approach, in order to show the differences with respect to the
(partial and fully) degenerate cases which will follow.

\paragraph{Order 1:}
For the sake of simplicity we may chose $(\theta_1, \theta_2,
\theta_3, \theta_4)=(0, 0, 0, 0)$. By inserting the zeroth order of
the solution \eqref{e.4Dtorus} into \eqref{e.order.1.S} and
\eqref{e.order.1.notS}, for the specific choice of $\theta_j$, we get
e.g. for $j\in S$
$$-2\omega
\Re(\phi_j^{(1)})=R\Leftrightarrow\Re(\phi_j^{(1)})=-\frac{2}{3R}.$$
In a similar way we get the first order of the solution for the rest
of the sites. The results for all the sites are
\begin{equation*}
\begin{array}{l || cccc|c|cccc }
j & \cdots & -2 & -1 & 0 & j\in S & 5 & 6 & 7 & \cdots
\\
\hline
-\phi_j^{(1)} &
0 & \frac2{3R} & \frac2{3R} & \frac4{3R}
& \frac2{3R} - \im \alpha_j\
& \frac4{3R} & \frac2{3R} & \frac2{3R} & 0
\end{array},
\end{equation*}
where the four $\alpha_j\in\RR$
represent the four independent kernel directions. The conservation
law \eqref{e.current.o1} is sastisfied, provided the $\alpha_j$
fulfill the following set of linear homogeneous equations
\begin{displaymath}
\begin{cases}
2\alpha_1-\alpha_2-\alpha_4= 0\\
\alpha_1+\alpha_2-\alpha_3-\alpha_4=0\\
\alpha_1+\alpha_3-2\alpha_4=0
\end{cases}
\qquad\Longrightarrow\qquad
\alpha_1=\alpha_2=\alpha_3=\alpha_4\ .
\end{displaymath}
Thus, the first order expansion of the density current removes all the
kernel directions but one, as expected by the non-degeneracy. The
remaining kernel direction ${\bf \alpha}=\alpha_1(1,1,1,1)$ represents
the effect of the rotational symmetry, which cannot be removed from the
perturbation expansion.

\paragraph{Order 2:} As before, using now 
\eqref{e.order.2.S} and \eqref{e.order.2.notS}, one has (limiting to
the sites $-2\leq j\leq 7$)
\begin{equation*}
\begin{array}{l || ccc|c|ccc }
j  & -2 & -1 & 0 & j\in S & 5 & 6 & 7
\\
\hline
\rule{0pt}{3.5ex}
\phi_j^{(2)} 
& -\frac{(8+6\im\alpha_1 R)}{9R^3}
& -\frac{2\im\alpha_1}{3R^2} & -\frac{(20+12\im\alpha_1 R)}{9R^3}
& -\frac{(16+9\im\alpha_j^2 R^2)}{18R^3} + \im\alpha_j'
& -\frac{(20+12\im\alpha_1 R)}{9R^3} & -\frac{2\im\alpha_1}{3R^2} 
& -\frac{(8+6\im\alpha_1 R)}{9R^3}
\end{array}
\end{equation*}

If we insert the second order corrections in the conservation law
\eqref{e.current.o2}, we get the same system of linear homogeneous
equations obtained at first order in the new kernel variables $\alpha_j'$
\begin{displaymath}
\begin{cases}
2\alpha_1'-\alpha_2'-\alpha_4'= 0\\
\alpha_1'+\alpha_2'-\alpha_3'-\alpha_4'=0\\
\alpha_1'+\alpha_3'-2\alpha_4'=0
\end{cases}
\qquad\Longrightarrow\qquad
\alpha'_1=\alpha'_2=\alpha'_3=\alpha'_4\ ,
\end{displaymath}
whose solution is independent of $\alpha_1$ and again leaves the
symmetry direction ${\bf \alpha'}=\alpha'_1(1,1,1,1)$ as the only
kernel direction.

Although it is not feasible to proceed explicitly in the expansion,
in this non-degenerate case it is easy to figure out that, at any
order, the conservation law would produce always the same system of
linear homogeneous equation in the new variables. Thus one can
recursively and uniquely determine all the needed coefficients,
leaving as free parameter only the gauge direction, as expected
from the symmetry of the system.

\begin{remark}
{\bf The $(\pi, \pi, \pi)$ family.} By using similar arguments as with
the $(0, 0, 0)$ family we also conclude that this family is also
non-degenerate and thus it can be iteratively constructed order by
order.

\end{remark}


\subsection{The asymmetric $(\varphi,\pi-\varphi,\pi-\varphi)$ and
  $(\vphi,\vphi,\pi-\vphi)$ vortex solutions}
\label{ss:asym3}

We now deal with the problem of continuing any of the phase-shift
solutions of the first two families of \eqref{e.sol-shifts}, except
$\pm(\frac{\pi}{2},\frac{\pi}{2},\frac{\pi}{2})$.  Concerning the
first one the perturbation expansion goes as follows.

\paragraph{Order 1:}
As before we represent the solution in the following table
\vskip-12pt
\renewcommand*{\arraystretch}{1.5}
\begin{equation*}
\overbrace{
  \begin{array}{|c|c|c|c|}
    1 & 2 & 3 & 4
    \\
    \hline
      -\frac4{3R} + \frac2{3R}\cos(\vphi) + \im \alpha_1
    & -\frac{4}{3R} e^{\im \vphi} + \im \alpha_2 e^{\im \vphi}
    & \frac{4}{3R} +\frac2{3R}\cos(\vphi) - \im \alpha_3
    & -\frac{4}{3R} e^{-\im \vphi} + \im \alpha_4 e^{-\im \vphi}
  \end{array}
}
^{\hskip18pt
  \begin{array}{l || ccc|c|ccc }
    j & -2 & -1 & 0 & \cdots & 5 & 6 & 7
    \\
    \hline
    \phi_j^{(1)}
    & -\frac2{3R} & -\frac2{3R} e^{\im \vphi} & 0
    & \cdots
    & -\frac4{3R}\cos(\vphi) & \frac2{3R} & -\frac2{3R} e^{-\im \vphi}
  \end{array}
}
\end{equation*}
\renewcommand*{\arraystretch}{1}

The conservation law \eqref{e.current.o1} is satisfied, provided
$\alpha_j$ fulfill 
\vskip-12pt
\begin{equation}
  \label{e.M1.fam3}
  M\cdot \pmb{\alpha}=0 \ ,
\qquad
  M:=\cos(\vphi)
  \begin{bmatrix}
    2 & -1 &  0 & -1\cr
    1 & -1 &  1 & -1 \cr
    1 &  0 & -1 &  0 \cr
  \end{bmatrix} \ ,
\quad 
  \pmb{\alpha}:=\begin{pmatrix}
  \alpha_1\\\alpha_2\\\alpha_3\\\alpha_4
  \end{pmatrix} \ .
\end{equation}
It is immediate to notice that for $\vphi=\pm\frac{\pi}2$ the system
is identically satisfied. This is the effect of the full degeneracy of
the vortex solutions, that will be treated separately later. So,
assuming $\vphi\not=\pm\frac{\pi}2$, we get
\begin{displaymath}
  \alpha_3=\alpha_1 \ ,
\quad
  \alpha_4=2\alpha_1 - \alpha_2 \ ,
\end{displaymath}
which leaves two Kernel directions in the problem: the gauge direction
and the tangent direction to the $\vphi$-family.

\paragraph{Order 2:} The expansion now proceeds as in the previous example,
computing the second order corrections $\phi_j^{(2)}$. For $j\in S$ in
particular the solutions can be taken in the form
\begin{displaymath}
  \phi_j^{(2)} = u_j^{(2)}(\alpha_1,\alpha_2)e^{\im\theta_j} + 
    \im \alpha'_j e^{\im\theta_j} \ ,
\end{displaymath}
with four new kernel directions $\alpha_j'$, as in the previous
example. We omit listing here the explicit expressions; once inserted
them into \eqref{e.current.o2}, we are left with three linear
equations, in this case nonhomogeneous, corresponding again to
$J^{(2)}_{2,3,4}=0$, which take the form
\begin{equation}
  \label{e.M2.fam3}
  M\cdot \pmb{\alpha'} = {\bf a} + {\bf b}(\pmb{\alpha}) \ ,
\qquad
  {\bf a}:= \frac{8}{9R^3}\sin(\vphi)
    \begin{pmatrix} 1\cr 1\cr 1\cr\end{pmatrix}\ ,
\quad
  {\bf b}(\pmb{\alpha}) := \frac4{3R^2}\cos(\vphi)^2\alpha_1
    \begin{pmatrix} 1\cr 0\cr 1\cr\end{pmatrix}\ .
\end{equation}
The rank of $M$ is equal to the rank of the augmented matrix $M|{\bf
  a} + {\bf b}(\pmb{\alpha})$ only if $\sin(\vphi) = 0$,
i.e. $\vphi=\graff{0,\pi}$, hence any possible $\C^2$ asymmetric
vortex solution in this family is excluded and only trivial phase
families are allowed.

Taking into account the second family, all the calculations are
essentially the same except from some permutation of indices, i.e. in
this case the matrix $M$ is
\begin{equation}
  \label{e.M1.fam1}
  M:=\cos(\vphi)
  \begin{bmatrix}
    0 & 1 & 0 & -1\cr
    1 & -1 & 1 & -1 \cr
    1 & 0 & 1 & -2 \cr
  \end{bmatrix} \ ,
\end{equation}
and ${\bf a}$ has a change in the sign, but the conclusions
are the same as before.


\subsection{The asymmetric $(\vphi,\pi-\vphi,\vphi)$ vortex solution}
\label{ss:asym2}

We proceed dealing with the problem of continuing any of the
phase-shift solutions of the last family of \eqref{e.sol-shifts},
again with the exclusion of the solutions
$\pm(\frac{\pi}{2},\frac{\pi}{2},\frac{\pi}{2})$. With the convenient
choice of ${\bf \theta} = \tond{0,\vphi,\pi,\pi+\vphi}$, the
perturbation expansion goes as follows.

\paragraph{Order 1:}
\vskip-12pt
\renewcommand*{\arraystretch}{1.5}
\begin{equation*}
\overbrace{
  \begin{array}{|c|c|c|c|}
    1 & 2 & 3 & 4
    \\
    \hline
      -\frac4{3R} + \im \alpha_1
    & -\frac{4}{3R} e^{\im \vphi} + \im \alpha_2 e^{\im \vphi}
    & \frac{4}{3R} - \im \alpha_3
    & \frac{4}{3R} e^{\im \vphi} - \im \alpha_4 e^{\im \vphi}
  \end{array}
}
^{\hskip-46pt
  \begin{array}{l || ccc|c|ccc }
    j & -2 & -1 & 0 & \cdots & 5 & 6 & 7
    \\
    \hline
    \phi_j^{(1)}
    & -\frac2{3R} & -\frac2{3R} e^{\im \vphi} & 0
    & \cdots
    & 0 & \frac2{3R} & \frac2{3R} e^{\im \vphi}
  \end{array}
}
\end{equation*}
\renewcommand*{\arraystretch}{1}
The conservation law \eqref{e.current.o1} is satisfied, provided
$\alpha_j$ fulfill
\begin{equation}
  \label{e.M1.fam2}
  M\cdot \pmb{\alpha}=0 \ ,
\qquad
  M:=\cos(\vphi)
  \begin{bmatrix}
    0 & 1 & 0 & -1\cr
    1 & 1 & -1 & -1 \cr
    1 & 0 & -1 & 0 \cr
  \end{bmatrix}\ .
\end{equation}
It is immediate to notice that for $\vphi=\pm\frac{\pi}2$ the system
is again identically satisfied, as a result of the full degeneracy of
the vortex solutions. So, assuming $\vphi\not=\pm\frac{\pi}2$, we get
\begin{displaymath}
\alpha_3= \alpha_1\ ,\quad
\alpha_4= \alpha_2\ .
\end{displaymath}

\paragraph{Order 2:}
As in the previous subsection we omit the explicit expression of the
solution $\phi_j^{(2)}$; we directly give the result of their use into
the density current conservation law \eqref{e.current.o2} which
reduces to the three equations corresponding to
$J_{2,3,4}^{(2)}=0$, which take the form
\begin{equation}
  \label{e.M2.fam2}
  M\cdot \pmb{\alpha'} = {\bf a} \ ,
\qquad
  {\bf a}:= -\frac{8}{9R^3}\sin(\vphi)
  \begin{pmatrix} 1\cr 1\cr 1\cr\end{pmatrix} \ .
\end{equation}
Again, the presence of the vector ${\bf a}$ implies $\sin(\vphi) = 0$,
hence $\vphi\in\graff{0,\pi}$, thus any possible $\C^2$ asymmetric
vortex solution in the third family is excluded.


\subsection{The vortex solutions}
\label{ss:vs}

The previous analysis has shown that only the phase shifts, i.e.,
$\vphi_j\in\graff{0,\pi}$, or the vortex solutions, i.e.,
$\vphi_j=\pm\frac\pi2$, have a chance to be continued as $\C^2$
solutions. To complete our analysis, we wish to
now show that also the vortex solutions cannot
be continued in $\eps$. This, however, requires a separate expansion
of the first two perturbation orders, due to the complete degeneracy
in~\eqref{e.M1.fam3} and
\eqref{e.M1.fam2}.

\paragraph{Order 1:} 

\noindent Let, for the sake of simplicity, $(\theta_1, \theta_2,
\theta_3, \theta_4)=(0, \pi/2, \pi, 3\pi/2)$ for $j\in S$; i.e.~the
equations for $\phi_j^{(1)}$ are in a form such that only their real
or imaginary parts appear, thus only half of the solution is
determined. For example, for $j=1$, remarking that $\phi^{(0)}_{2} +
\phi^{(0)}_{4}=0$, we have
\begin{equation*}
  -\omega\tond{\phi_1^{(1)} + \overline{\phi}_1^{(1)}} = 
  2R -\frac12\quadr{\phi^{(0)}_{2} + \phi^{(0)}_{4}}
\qquad\Longrightarrow\qquad
  \phi_1^{(1)} = -\frac4{3R}+\im\alpha_1\ ;
\end{equation*}
for the remaining values of $j$ the solution is completely
determined. As a result we have

\renewcommand*{\arraystretch}{1.5}
\begin{equation*}
  \begin{array}{l || ccc|c|c|c|c|ccc }
    j & -2 & -1 & 0 & 1 & 2 & 3 & 4 & 5 & 6 & 7
    \\
    \hline
    \phi_j^{(1)}
    & -\frac2{3R} & -\frac{2\im}{3R} & 0
    & -\frac4{3R} + \im \alpha_1
    & \alpha_2 - \im\frac{4}{3R}
    & \frac{4}{3R} + \im \alpha_3
    & \alpha_4 + \im \frac{4}{3R}
    & 0 & \frac2{3R} & \frac{2\im}{3R}
  \end{array}
\end{equation*}
\renewcommand*{\arraystretch}{1}

\begin{remark}
The four free parameters $\alpha_j$ represent the first order
expansion in $\eps$ of the kernel directions in the forthcoming
Lyapunov-Schmidt decomposition. Indeed, at this order, the equations
we are able to solve represent the range equations in the same
decomposition.
\end{remark}

In this case, and in contrast with the previous one, the density
current equations $J_j^{(1)}=0$ do not provide any further information
on the first order solutions $\phi_j^{(1)}$, with $j\in S$. Indeed,
the three equations $J_{2,3,4}^{(1)}=0$ give the trivial system
\begin{equation}
  \label{e.M1.vort}
  M \pmb{\alpha} = 0 \ , \qquad\text{with}\quad M:=0 \ ,
\end{equation}
so that the unknown $\alpha_j$ remain undetermined (see
also~\eqref{e.M1.fam3}, \eqref{e.M1.fam1} and~\eqref{e.M1.fam2} with
$\varphi=\pm\frac{\pi}{2}$).

\paragraph{Order 2:}
We face a similar situation as before, with the $\phi_j^{(2)}$, for
$j\in S$, appearing into the equations only through their real or
imaginary parts: thus we are left with 4 new parameters $\alpha_j'$,
corresponding to the real part of $\phi_{2,4}^{(2)}$ and to the
imaginary part of $\phi_{1,3}^{(2)}$, which are not determined, as
before. The other four components appear instead as functions of the
previous four free parameters $\alpha_j$. As before, for $j\not\in S$
no issues arise. To summarize, defining
\begin{equation*}
  f(x,y):=\frac{x+y}{3R^2}\ ,
\qquad
  g(x):=\frac{10}{9R^3} + \frac{x^2}{2R} \ .
\end{equation*}
and factorizing some constants in the values of the external sites, we
have

\renewcommand*{\arraystretch}{1.5}
\begin{multline*}
  \begin{array}{l || cccccc| }
    j & -5 & -4 & -3 & -2 & -1 & 0 
    \\
    \hline
    \rule{0pt}{3.5ex}
    \frac94R^3\phi_j^{(2)}
    & 1 & \im & 1
    & -2 +\im\tond{1-\frac{3R\alpha_1}2} & \tond{1-\frac{3R\alpha_2}2} -2\im
       & \im\tond{1-\frac{3R(\alpha_1+\alpha_3)}2}
  \end{array}
\\
  \begin{array}{l||c|c|}
    j & 1 & 2
    \\
    \hline
    \phi_j^{(2)}
    & \tond{f(\alpha_2,\alpha_4)-g(\alpha_1)} + \im \alpha'_1
    & \alpha'_2 + \im \tond{f(\alpha_1,\alpha_3)-g(\alpha_2)}
  \end{array}\phantom{XXXX}
\\
  \begin{array}{l||c|c|}
    j & 3 & 4
    \\
    \hline
    \phi_j^{(2)}
    & \tond{f(\alpha_2,\alpha_4)+g(\alpha_3)} + \im \alpha'_3
    & \alpha'_4 + \im \tond{f(\alpha_1,\alpha_3)+g(\alpha_4)}
  \end{array}
\\
  \begin{array}{l || ccc }
    j & 5 & 6 & 7
    \\
    \hline
    \rule{0pt}{3.5ex}
    \frac94R^3\phi_j^{(2)}
    & -\tond{1+\frac{3R(\alpha_2+\alpha_4)}2}
    & 2 -\im\tond{1+\frac{3R\alpha_3}2}
    & -\tond{1+\frac{3R\alpha_4}2} +2\im
  \end{array}
\end{multline*}
\renewcommand*{\arraystretch}{1}

We consider now the second order of the expansion of the conservation law
\eqref{e.current}. It turns out that once again all the equations $
J_n^{(2)} = 0$, for $n\not\in\graff{2,3,4}$, are identically
satisfied, providing no information on any of the eight free
parameters $\pmb{\alpha}$ and $\pmb{\alpha'}$. The only non trivial
equations are again
\begin{displaymath}
J_2^{(2)} = 0\ ,\quad J_3^{(2)} = 0\ ,\quad J_4^{(2)} = 0\ .
\end{displaymath}
After long and tedious manipulation that we here omit, one reaches the
following system of three quadratic equations which remarkably depend
only on the first order kernel variables $\alpha_j$, and not on the
second order variables $\alpha_j'$

\begin{equation}
  \label{e.current.o2.vort}
  \begin{cases}
    C = (b + d)(b - d + 2a) \ ,\\ 
    C = (a + c)(c - a + 2d) \ ,\\ 
    C = b^2 - d^2 + c^2 - a^2 + 2(ad - bc) \ ,
  \end{cases}
\quad\text{where}\quad
C:= \frac{16}{9R^2} \ ,
\quad
\begin{aligned}
  a&\equiv\alpha_1 \ , &b\equiv\alpha_2
  \\
  c&\equiv\alpha_3 \ , &d\equiv\alpha_4
\end{aligned}
  \ .
\end{equation}
The above system, by taking the differences of its equations, assume
the form
\begin{equation*}
\begin{cases}
  0 = (b + d)(d - b + 2c) \ ,\\ 
  0 = (a + c)(a - c + 2b) \ ,\\ 
  0 = (a + b)^2 - (c + d)^2 \ .
\end{cases}
\end{equation*}
The third equation implies either $a+b=c+d$ or $a+b=-(c+d)$, call them
case 1) and 2) respectively. Moreover in the first two equations one
has to exclude $b+d=0$ and $a+c=0$ because those would imply $C=0$
in~\eqref{e.current.o2.vort}. Thus necessarily we have
\begin{equation*}
\begin{cases}
  b - d = 2c \ ,\\ 
  c - a = 2b \ ,
\end{cases}
\quad\Longrightarrow\quad
\begin{cases}
  b + c = -(a + d) \ ,\\ 
  a + b - (c + d) = 2(c - b) \ .
\end{cases}
\end{equation*}
In case 1), the first on the right imply $a+b=c+d=0$, so that using
the second we have $b=c$, which gives $a=d$, so that we end up with
$b+d=a+c=0$, which we already excluded above since it is equivalent to
$C=0$. We are thus left with case 2). Here the contradiction is
obtained putting the system on the left in one of the first two
equations of~\eqref{e.current.o2.vort} to get $C=2(a+c)(b+d)$. Since
in case 2) we have $a+c=-(b+d)$ we would get $C<0$.
\\
\\
To summarize the system turns out to be impossible, thus proving the
incompatibility of the vortex solution.

\begin{remark}
After concluding the presentation of the perturbative approach of the
problem, the reader could assume that this methodology is only
applicable to the present system because of the uniquenss of the form
of \eqref{e.current}. Actually, this is not true, since this approach
can be very easily adapted in order to be used in every system of the
dNLS type with interactions beyond the nearest-neighbor ones and the
corresponding density current can be easily calculated. In fact the
method has been used in the zigzag ~\cite{PKSPK17kg} configuration,
where the corresponding vortex configurations also are excluded: in
this case the degeneracy occurs, however, in lower order than the
present example. Anyway, although the density current method is
applicable to any kind of linear interaction, the kind of degeneracy
is highly dependent on the number of sites considered in the
configuration.
\end{remark}


\section{Proof of Theorem~\ref{t.nonexistence} via
  Lyapunov-Schmidt decomposition}
\label{s:LS}

In the present Section we give the proof of
Theorem~\ref{t.nonexistence}, without the $C^2$ regularity restriction
of the previous section.  Consider again the stationary equation
\eqref{e.dNLS.phi}, now written as
\begin{equation}
  \label{e.dNLS.abst}
  \omega \phi + \eps L\phi - \frac34\phi|\phi|^2 = 0\ ,
\qquad\text{with}\quad
  L:=\frac12\quadr{\Delta_1+\Delta_3}\ ,
\end{equation}
and $\phi\in\ell^2(\CC)$. Since we are interested in the continuation
of vortex-like solutions from the anti-continuum limit, we introduce
the translation
\begin{equation}
  \label{e.dec.phi.1}
  w(\eps) := \phi(\eps) - v \ ,
\end{equation}
where $v$ is a fixed solution of the unperturbed equation, as in
formula~\eqref{e.4Dtorus}, and $w(\eps)$ represents a small
  displacement around it, namely
\begin{equation}
  \label{e.norm.w}
  v_j =
  \begin{cases}
    R e^{\im\theta_j} & j\in S\ ,
\\
    0              & j\not\in S\ ,
  \end{cases}
\qquad\text{and}\quad
  \norm{w}\ll \norm{v} \ .
\end{equation}

\begin{remark}
  \label{r.1}
  The change of coordinate \eqref{e.dec.phi.1} is one of the places
  where a key difference appears with respect to the previous
  section. Although it might be seen as the analogue of the
  decomposition~\eqref{e.exp.phi} --- with $v$ taking in the present
  Section the place of $\phi^{(0)}$, and $w$ the role of the higher
  order terms of such an expansion --- at variance
  with~\eqref{e.exp.phi}, in~\eqref{e.dec.phi.1} no regularity is
  assumed (apart from the obvious continuity), as $w$ is simply a
  small displacement. We will instead exploit the regularity at the
  level of the equations.
\end{remark}


Exploiting~(i) that $v$ is a generic solution of the unperturbed
problem and~(ii) that $w$ is small in $\epsilon$, we
split~\eqref{e.dNLS.abst} in powers of $w$ and define the following
four functions
\begin{equation}
  \label{e.F.1}
  \begin{aligned}
    F(v;w,\eps) := \eps Lv +
      \bigl( & \tikzmark{Lambda1}\Lambda w + \eps Lw \bigr) - 
      \bigl( \tikzmark{N21}N_2(v;w) + \tikzmark{N31}N_3(w) \bigr) \ ,
  \\
    \noalign{\vskip1pt}
    &\phantom{\Lambda w + \eps Lw \bigr) - \bigl(N_2(v;w) + \null}
    \tikzmark{N32}N_3(w) := \frac34|w|^2 w
  \\
    &\phantom{\Lambda w + \eps Lw \bigr) - \bigl(}
    \tikzmark{N22}N_2(v;w) := \frac34\tond{w^2\overline v + 2v|w|^2}
  \\
    & \tikzmark{Lambda2}\Lambda w :=
      \omega w - \frac34\tond{v^2\overline{w} + 2|v|^2 w}\ .
  \end{aligned} 
  \begin{tikzpicture}[overlay,remember picture]
    \node (Lam1) [below of = Lambda1, node distance = .1 em, anchor=west]{};
    \node (Lam2) [above of = Lambda2, node distance = .8 em, anchor=west]{};
    \draw[->,dotted,in=90, out=-90] (Lam1) to (Lam2);
    \node (NN21) [below of = N21, node distance = .1 em, anchor=west]{};
    \node (NN22) [above of = N22, node distance = .8 em, anchor=west]{};
    \draw[->,dotted,in=90, out=-90] (NN21) to (NN22);
    \node (NN31) [below of = N31, node distance = 0. em, anchor=west]{};
    \node (NN32) [above of = N32, node distance = .6 em, anchor=west]{};
    \draw[->,dotted,in=90, out=-90] (NN31) to (NN32);
  \end{tikzpicture}
\end{equation}
Thus~\eqref{e.dNLS.abst} takes the form
\begin{equation}
  \label{e.dNLS.dec.1}
  F(v;w(\eps),\eps)= 0\ .
\end{equation}
The linear part has been split into two terms since 
\begin{equation}
  \label{e.lambda}
  \Lambda = \bigl(D_w F\bigr) (v;0,0) \ ,
\end{equation}
being independent of $\epsilon$, is the operator one has to look at for the continuation of $v$ as
a solution $F(v;0,0)=0$. It is useful, considering the shape of $v$
given in \eqref{e.norm.w}, to represent $\Lambda$ in a matrix form
\begin{equation}
  \label{e.lambda.matrix}
  \Lambda = \begin{pmatrix}
    \omega\Id & 0 & 0   \\
    0 & -\omega M_S & 0 \\
    0 & 0 & \omega \Id  \\
  \end{pmatrix}
\end{equation}
where $M_S$ is a block matrix, composed of 4 blocks $M_{S,j}$, each
in the form
\begin{equation}
  \label{e.MS}
  M_{S,j} := \begin{pmatrix}
    2\cos^2(\theta_j) & \sin(2\theta_j)\\
    \sin(2\theta_j) & 2\sin^2(\theta_j) \\
  \end{pmatrix}\ .
\end{equation}
Each block $M_{S,j}$ has zero determinant, its one dimensional kernel
being given by
\begin{equation}
  \label{e.ej}
  Ker\tond{M_{S,j}} = \inter{e_j}\ ,\qquad e_j := 
  \begin{pmatrix}
    -\sin(\theta_j)  \\
    \cos(\theta_j) \\
  \end{pmatrix} \equiv \im e^{\im \theta_j}\ ;
\end{equation}
hence, the differential $\Lambda$ has a four dimensional kernel, as
expected by the four dimensional tangent space of $\mathbb{T}^4$,
which is given by
\begin{displaymath}
  Ker\tond{\Lambda} = \inter{f_1,f_2,f_3,f_4}\ ,
\end{displaymath}
where each $f_j$ is the embedding of the corresponding $e_j$ in the
$2N$ dimensional real phase space (or in the $N$ dimensional complex
phase space). Hence $\Lambda$ is not invertible.

\subsection{Lyapunov-Schmidt decomposition}

Given the above consideration on the operator $\Lambda$, instead of
the standard implicit function theorem approach, we need to perform a
Lyapunov-Schmidt decomposition.

We denote the kernel of $\Lambda$ as $K$, and its range as $H$, and
$\Pi_K$ and $\Pi_H$ the corresponding projectors.  Due to the
structure of the phase space\footnote{This is trivial in the finite
  dimensional case; in the infinite dimensional case, one has to
  notice that $F:\ell^2\rightarrow\ell^2$.} $\Ph$, we can identify
$H=K^\perp$ and $\Ph = K\oplus H$. We consider the decomposition
\begin{displaymath}
  w = k+h \ ,
\qquad\hbox{with}\quad
  k\in K 
\quad\hbox{and}\quad
  h\in H \ ,
\end{displaymath}
and we project \eqref{e.dNLS.dec.1} onto the two spaces
\begin{align*}
  F_H(v;h,k,\eps):=\Pi_HF\tond{v;h+k,\eps} = 0 \ ,
\\
  F_K(v;h,k,\eps):=\Pi_KF\tond{v;h+k,\eps} = 0 \ .
\end{align*}

The strategy is, as usual, to first solve the range equation,
exploiting the inveritbility of $\Lambda$ on $H$, so as to find an
$h(v;k,\eps)$ such that $F_H(v;h(v;k,\eps),k,\eps)=0$; then to insert
such a solution into the kernel equation.

\begin{remark}
  As anticipated in Remark~\ref{r.1}, we stress once more that here we
  are not assuming any regularity of the solution. We will
  nevertheless expand $h$ since in the range equation $F_H=0$, the
  function $F_H$ is regular, and thus a solution obtained via the
  implicit function theorem preserves such a regularity. The possible
  non-regularity of the solution may take place in the dependence of
  $k$ on $\eps$ at the level of the kernel equation.
\end{remark}

We perform a preliminary simplification of the two equations above by
observing the following elementary facts. First of all, we clearly have
$\Lambda(h+k) = \Lambda h$; moreover the orthogonality of the kernel
and range subspaces implies that we have $\Re (k\bar h)=0$, simplifying
the nonlinear terms; and in particular, from~\eqref{e.ej}
and~\eqref{e.norm.w}, we get $\Re (e_j\bar v_j)=0$
for $j\in S$, which confirms that $v\in H$. Furthermore,
exploiting again that within $S$ one has that $v_j$ and $h_j$ are
directed as $e^{\im\theta_j}$ and $k_j$ is parallel to $\im
e^{\im\theta_j}$, it is easy to make the projections, e.g. as in $\bar
v \tond{h^2+k^2}=v\tond{|h|^2-|k|^2}\in H$ and $2\bar v hk=2\Re(v\bar
h)k\in K$. Thus we have

\begin{align}
\label{e.Range}
  \Lambda h + &\eps\Pi_H\bigl(L(h+k+v)\bigr)
  - \frac34 \Bigl( \bigl( 3|h|^2 + |k|^2 \bigr) v +
    \bigl( |h|^2 + |k|^2 \bigr) h \Bigr) = 0 \ ,
  \tag{R}
\\
\label{e.Kernel}
  &\eps\Pi_K\bigl(L(h+k+v)\bigr) -
  \frac34 \Bigl({|h|^2 + |k|^2 + 2\Re\bigl(v\bar h\bigr)}\Bigr) k = 0 \ .
  \tag{K}
\end{align}


\subsection{Range equation}

As we commented above, the solvability of the range equation is not an
issue.  The interesting point is instead to shed some light on the
general structure of the solution that we will exploit later
on. According to such a purpose, and recalling that $h(v;k,\eps)$ is
regular in $k$ and $\eps$, we start by expanding it as follows
\begin{equation}
  \label{e.h.eps}
  h = h^{(0)}(v;k) + \eps h^{(1)}(v;k) + \eps^2h^{(2)}(v;k) + 
    \tilde h^{(3)}(v;k,\eps)\ .
\end{equation}
Inserting the above equation in~\eqref{e.Range}, we can produce explicit
expressions for the above expansion by solving the range equation
iteratively. Before proceeding into such a task, it is also useful to
notice here that the local terms can be further simplified splitting
between core sites ($j\in S$) and the other ones: in particular the
operator $\Lambda$, once applied to an element of the range $h\in H$,
simply becomes
\begin{displaymath}
\Lambda h = \begin{cases}
  \phantom{-2}\omega h     & j\not\in S\ ,  \\
  -2          \omega h     & j\in S\ ,      
\end{cases} 
\end{displaymath}
and the nonlinear part of the equation takes the form
\begin{displaymath}
\Pi_H(N_2+N_3)=\frac34 \begin{cases}
  |h|^2 h                                     & j\not\in S \ ,   \\
  \tond{|h|^2 + |k|^2}\tond{v+h} + 2|h|^2v     & j\in S \ .
\end{cases} 
\end{displaymath}


\paragraph{Order 0:} The order 0 component of the range equation takes the
form

\begin{equation*}
  \Lambda h^{(0)} - \frac34 \tond{ \bigl({\bigl|{h^{(0)}}\bigr|^2 +
      |k|^2}\bigr)\bigl({v + h^{(0)}}\bigr) + 2 \bigl|{h^{(0)}}\bigr|^2 v } = 0 \ ,
\end{equation*} 
which can be split, after further simplifications, and introducing for
the core sites ($j\in S$) the temporary notation $h_j^{(0)}=\ell
e^{\im\theta_j}$, with $\ell\in\RR$, as

\begin{align}
\notag
  h^{(0)} \tond{R^2 - \ell^2 } &= 0      & j\not\in S,
\\
\label{e.Range.0}
  \bigl({v + h^{(0)}}\bigr) \tond{\ell^2 + |k|^2 + 2R\ell}
    &= 0                   & j\in S.
\end{align}
Recalling the smallness condition~\eqref{e.norm.w}, we get
\begin{align*}
  h^{(0)}&=0                                      & j\not\in S,
\\
  \ell&=\sqrt{R^2-|k|^2}-R             & j\in S,
\end{align*}
so that, introducing an expansion in $|k|$, we may conveniently write

\begin{equation}
  \label{e.h.02}
\begin{aligned}
  h^{(0)}(v;k) = h^{(0,2)} & + \mathcal{O}(|k|^3)\ , &
\\
  h^{(0,2)} &= -\frac{3/4|k|^2v}{2\omega} = 
             \frac34\Lambda^{-1}\tond{ |k|^2 v}   &\quad j\in S\ .
\end{aligned}
\end{equation}

\paragraph{Order 1:} At order 1 we have

\begin{equation}
  \label{e.Range.1}
  \Lambda h^{(1)} + \Pi_H L \bigl({h^{(0)} +  k +  v}\bigr)
  -\frac34\tond{ 2\Re\bigl({h^{(0)}\overline{h}^{(1)}}\bigr)\bigl({3v + h^{(0)}}\bigr)
    +\bigl({\bigl|{h^{(0)}}\bigr|^2 + |k|^2}\bigr) h^{(1)} } = 0 \ .
\end{equation}
We now expand $h^{(1)}$ up to order one in the variables $k$ , namely
\begin{displaymath}
  h^{(1)}(v;k) = h^{(1,0)}(v) + h^{(1,1)}(v;k) + \mathcal{O}(|k|^2) \ ;
\end{displaymath}
replacing $h^{(1)}$ in~\eqref{e.Range.1} with the latter expansion and
recalling that we already obtained that $h^{(0)}$ is zero up to order
two in $k$, we are left with the following equations,
\begin{align}
  \label{e.h.1.0}
  \Lambda h^{(1,0)} + \Pi_HLv &= 0 
  \qquad\Longrightarrow\qquad
  h^{(1,0)} = -\Lambda^{-1}\Pi_HLv \ ;
\\
  \label{e.h.1.1}
  \Lambda h^{(1,1)} + \Pi_HLk &= 0 
  \qquad\Longrightarrow\qquad
  h^{(1,1)} = -\Lambda^{-1}\Pi_HLk \ .
\end{align}

\begin{remark}
The leading order solution $h^{(1,0)}(v)$ is nothing but the first
term $\phi^{(1)}$ obtained as solution of \eqref{e.order.1.S} in the
first part of the paper, with the exception of the kernel directions
$\alpha_j$ that cannot appear in the range part.
\end{remark}

\paragraph{Order 2:} The expansion of the range equation at order two
in $\eps$ gives
\begin{equation}
\label{e.Range.2}
\begin{aligned}
\Lambda h^{(2)} + \Pi_H L h^{(1)}
  -\frac34\biggl( 
    \tond{\bigl|{h^{(1)}}\bigr|^2 + 2\Re\bigl({h^{(0)}\overline{h}^{(2)}}\bigr)}&
      \tond{3v + h^{(0)}} 
\\
    + 2\Re\bigl({h^{(0)}\overline{h}^{(1)}}\bigr) &h^{(1)} + \bigl({\bigl|{h^{(0)}}\bigr|^2 + |k|^2}\bigr)  h^{(2)}
  \biggr) = 0 \ .
\end{aligned}
\end{equation}
Here, we need only the leading order of the expansion of $h^{(2)}$ in
powers of $k$, namely
\begin{equation}
  \label{e.h.2.taylor}
  h^{(2)} = h^{(2,0)}(v) + \mathcal{O}(|k|) \ ;
\end{equation}
and recalling that
\begin{displaymath}
  h^{(0)} = \mathcal{O}(|k|^2) \ ,
\qquad
  h^{(1)} = h^{(1,0)}(v) + \mathcal{O}(|k|) \ ,
\end{displaymath}
from \eqref{e.Range.2} we obtain, at leading order,
\begin{equation}
  \label{e.h.2.0}
  h^{(2,0)}(v) = -\Lambda^{-1} \tond{
    \Pi_H L h^{(1,0)} - \frac94 \bigl|{h^{(1,0)}}\bigr|^2 v
  } \ .
\end{equation}


\subsection{Kernel equation}

We consider now the kernel equation~\eqref{e.Kernel}, where we insert
the solution $h(v;k,\eps)$ of the range equation~\eqref{e.Range}, so
that we have to deal with
\begin{displaymath}
  F_K(v;k,\eps) := \Pi_K F(v;k,h(v;k,\eps),\eps)=0 \ ,
\qquad
  F_K(v;\cdot,\cdot):\RR^4\times\RR\to\RR^4 \ .
\end{displaymath}
We recall that, given a solution of such an equation for $\eps=0$, our
aim is to look for its continuation for $\eps\not=0$. In particular,
taking $v$ as in~\eqref{e.norm.w}, we require $w(\eps)$
as introduced in~\eqref{e.dec.phi.1} to solve~\eqref{e.dNLS.dec.1}
with the property $\lim_{\eps\to0}w(\eps)=0$ to guarantee that our
solution is indeed a continuation from $v$. The first relevant remark
is given by the following
\begin{lemma} Let $v\in\mathbb{T}^4$ and $k\in\RR^4$, then
\begin{equation}
  F_K(v;k,0)=0 \ .
\end{equation}
Thus, trivially, all the derivatives in $k$ vanish in zero, i.e.,
$D_k^{(m)}F_K(v;0,0)=0$.
\end{lemma}
\begin{proof}
Since the dependence on $\eps$ of $F_K$ is given also by the
corresponding dependence of $h(v;k,\eps)$, we insert the expansion
\eqref{e.h.eps} of $h$ in powers of $\eps$, and we perform the
corresponding expansion of the above function as $F_K(v;k,\eps) =
F_K(v;k,0) + {\mathcal O}(\eps)$. It turns out that
$F_K(v;k,0)=-\frac34\tond{\bigl|{h^{(0)}}\bigr|^2 + |k|^2 +
  2\Re\bigl({v\overline{h}^{(0)}}\bigr)}k$, which is identically zero
since $h^{(0)}$ solves~\eqref{e.Range.0}, and the thesis follows.
\end{proof}

From the above Lemma it follows that all $v$ as in~\eqref{e.norm.w}
represent possible candidates for the continuation, since we have
$F_K(v;0,0)=0$. Furthermore, for $\eps\not=0$ the kernel equation
takes the form
\begin{equation}
  \label{e.Kern_red}
  P(v;k,\eps)= 0 \ ,\qquad \eps P(v;k,\eps):= F_K(v;k,\eps)\ ,
\end{equation}
where obviously $P(v;\cdot,\cdot):\RR^4\times\RR\to\RR^4$, and in
$\eps=0$ one has to read the definition of $P$ as
$P(v;k,0)=\partial_\eps F_K(v;k,0)$.  Thus, we get a restriction on the
possible bifurcation points, as follows;

\begin{lemma}
A necessary condition for $v^*$ to be a bifurcation point of the kernel
equation is the following
\begin{equation}
\label{e.bif.point}
\Pi_K(Lv^*) = 0 \ .
\end{equation}
\end{lemma}

\begin{proof}
As we said, for $\eps\not=0$, the kernel equation is equivalent to
formula~\eqref{e.Kern_red}, so we want a $v^*$ and a $k(\eps)$ such
that $P(v^*;k(\eps),\eps)=0$, but with
$\lim_{\eps\to0}k(\eps)=0$. Thus, given the regularity of $P$, we need
\begin{equation}
  P(v^*,0,0) = 0 \ ,
\end{equation}
and using as before the expansion of $h$ and of $F_K$, one easily gets
\begin{equation}
  P(v,0,0) = \partial_\eps F_K(v,0,0) = \Pi_K(Lv) \ .
\end{equation}
\end{proof}

The compatibility condition~\eqref{e.bif.point} is equivalent to the
variational energy method introduced in \cite{Kap01} (see also Remark
2.1). Using the notation of the phase shifts $\varphi$
(see~\eqref{e.phidef}), one easily gets as candidate bifurcation
points the same ones shown in Section~\ref{s:zero.order}, i.e. the two
isolated points $\varphi\in\{(0,0,0),(\pi,\pi,\pi)\}$, the three
families~\eqref{e.sol-shifts}, and their intersections
$\varphi\in\{\pm(\frac{\pi}{2},\frac{\pi}{2},\frac{\pi}{2})\}$ that we
call symmetric vortices.

The next step is of course to test the applicability of the implicit
function theorem, i.e. to check whether $dP(v^*;0,0)$ has the correct
rank\footnote{It might be useful to remind that, with $P$ being regular in
  $(k,\eps)$, so will be the implicit function, if the corresponding
  theorem applies. Thus, if the dependence of $k$ with respect to
  $\eps$ in the implicit function is at least linear in the directions
  tranversal to the gauge, we will have $\text{rk} (D_k P(v^*,0,0)) =
  3$, and a regular $k(\eps)$. Nevertheless, in principle, we could have a
  regular implicit function with $\eps$ being a regular function of
  one of the $k$, but $k(\eps)$ sublinear and thus nonregular in the
  origin, as in the trivial lower dimensional example
  $f(y,\eps)=y^2-\eps=0$.}.
Introducing the linear operator
\begin{equation}
  \label{Toperator}
  T  k :=
    \Pi_K L k - \frac32\Re\left(v^*\overline{h}^{(1,0)}(v^*;0,0)\right)k  
  \ ,
\end{equation}
we have the following

\begin{lemma}
\label{l.nondegen.IFT}
Let $v^*$ satisfy~\eqref{e.bif.point}, then a sufficient
condition for its continuation for $\eps\not=0$ is
\begin{equation}
  \label{nondegen.IFT}
  {\rm rk} (T_1) = 3
\end{equation}
where the linear transformation $T_1$ is given by
\begin{equation}
  \label{nondegen.IFT.2}
  T_1 \begin{pmatrix} k \\ \eps  \end{pmatrix} :=
  \left(
    T k
\ \Big\vert\ 
    \Pi_K L h^{(1,0)}(v^*;0,0) \eps 
  \right)
  \ .
\end{equation}
\end{lemma}
\begin{proof}
Looking again at the expansion of $F_K$, one recognizes in
\eqref{Toperator} and \eqref{nondegen.IFT.2} the expressions of $D_k
P(v^*,0,0)[k]$ and $\partial_\eps P(v^*,0,0)\eps$, i.e. the
differential applied to its increment, so that
$T_1=dP(v^*,0,0):\RR^5\to\RR^4$, and condition~\eqref{nondegen.IFT} is
the standard one of the implicit function theorem, once we factor out
the zero eigenvalue that comes from the gauge invariance.
\end{proof}

If the above Lemma does not apply we enter in the field of the
degenerate implicit function theorems where a plethora of possible
subcases exist (see, e.g., \cite{Loud61}) without clear and easily
unifying statements. In particular if the rank is zero then clearly
the function $P$ starts with order two terms which have to be checked
explicitly. For the intermediate values of the rank, in general all
the possibilities are present, i.e. both to be able to directly prove existence, or non-existence, and the necessity to look at higher order
terms to overcome the degeneracy. There is nevertheless the easier
situation when the point we are interested in belongs to a family,
i.e. besides the gauge symmetry coming from the original problem, one
has that the $v^*$ solving~\eqref{e.bif.point} is not isolated, as for
our three families~\eqref{e.sol-shifts}. In such a case, a necessary
condition for the continuation can be given (see, e.g., Proposition
2.10 of~\cite{PelKF05b}). In the following Lemma we
report a formulation suitably adapted to our system. Let
us denote by $\R$ and $\K$ respectively the range and the kernel of
$T$, then we have

\begin{lemma}
\label{l.degen.IFT}
Let $v^*$ belong to a 2-dimensional family $v^*(\varphi)$
satisfying~\eqref{e.bif.point}. Assume that rk$(T) = 2$, $\Pi_K L
h^{(1,0)}(v^*(\varphi);0,0)\equiv0$ on the whole family and that
$\Pi_K L h^{(2,0)}(v^*;0,0)\not\equiv0$. A necessary condition for the
continuation of $v^*$ for $\eps\not=0$ is that
\begin{equation}
  \label{necessary.IFT}
  \Pi_K L h^{(2,0)}(v^*;0,0) \in \R
  \ .
\end{equation}
\end{lemma}

\begin{remark}
If the (complex) matrix $T$ is self-adjoint, then the
necessary condition \eqref{necessary.IFT} is equivalent to
\begin{equation}
  \label{necessary.IFT.2}
  \Pi_K L h^{(2,0)}(v^*;0,0) \perp \K
  \ .
\end{equation}
We will exploit this equivalence as a nonexistence argument for the
continuation of the three families.
\end{remark}

\begin{proof}
  We are interested in a solution of $P(v^*;k(\eps),\eps)=0$ as a
  continuation of $P(v^*;0,0)=0$.  Given that $v^*$ belong to a
  2-dimensional family (one dimension from the gauge, and the other
  being the proper family dimension), then $T=D_kP(v^*;0,0)$ has a
  zero eigenvalue of multiplicity at least 2. By the hypothesis on its
  rank, it follows that the multiplicity is exactly 2. The idea is
  clearly to perform a Lyapunov-Schmidt decomposition, so let us
  denote by $k_\K$ and $k_\R$ the kernel component of $k$ and its
  orthogonal one, respectively, and by $[\,\cdot\,]_\R$ and
  $[\,\cdot\,]_\K$ the projections onto $\R$
  and its orthogonal space, respectively. Our problem is
  then written as
  \begin{equation}
    \label{e.secondaryLS}
    \begin{aligned}
      \Bigl[ P(v^*;k_\R + k_\K,\eps) \Bigr]_\R = 0
      \\
      \Bigl[ P(v^*;k_\R + k_\K,\eps) \Bigr]_\K = 0
    \end{aligned}
  \end{equation}
  It is also useful to expand $P$ with respect to $k$ and $\eps$, namely
  \begin{equation}
    \label{e.Pexpansion}
    \begin{matrix}
      P(v^*;0,0)  & + &&&&& \phantom{\vdots}
      \\
      D_k P(v^*;0,0)[k]  & + & \eps \partial_\eps P(v^*,0,0)  &&&&
      \phantom{\vdots}
      \\
      D_{kk} P(v^*;0,0)[kk] & + & \eps D_k \partial_\eps P(v^*,0,0)[k]
      & +  & \eps^2 \partial^2_\eps P(v^*,0,0) && \phantom{\vdots}
      \\
      \vdots & + & \vdots & + & \vdots & + & \ddots
    \end{matrix}
  \end{equation}
The $n$-th column contains homogeneous terms of order $\eps^n$ and the
$r$-th line gather the homogeneous terms of order $r$ in $(k,\eps)$.
In the present case, $P(v^*;0,0)=0$ since $v^*$ is the solution we are
going to continue, $\partial_\eps P(v^*,0,0)=\Pi_K L
h^{(1,0)}(v^*;0,0)=0$ by hypothesis, $D_k P(v^*;0,0)=T$ so that
$Tk_\K=0$. Moreover, since $P(v^*(\varphi);0,0)=0$ on the whole
family, expanding $k=k_\K+k_\R$, each term in the first column
vanishes when applied only to $k_\K$ (which is characterized exactly
by the family directions due the hypothesis rk$(T) = 2$); and the same
applies to the second column, since $\partial_\eps
P(v^*(\varphi),0,0)=0$ on the whole family too. Thus the kernel
directions of $T$ are kernel directions for $D_k \partial_\eps
P(v^*,0,0)$ and also for the higher order terms in the $k$
expansion. We thus have
  \begin{equation}
    \label{e.Pexpansion2}
    \begin{matrix}
      0  & + &&&&& \phantom{\vdots}
      \\
      T k_\R  & + & 0  &&&&
      \phantom{\vdots}
      \\
      B(k_\R,k_\R)+B(k_\R,k_\K) & + & \eps M k_\R
      & +  & \eps^2 \partial^2_\eps P(v^*,0,0) && \phantom{\vdots}
      \\
      \vdots & + & \vdots & + & \vdots & + & \ddots
    \end{matrix}
  \end{equation}
  where $B$ denotes the bilinear form $D_{kk} P(v^*;0,0)$ and $M=D_k
  \partial_\eps P(v^*,0,0)$. Recalling that $k$ has to vanish with
  $\eps$, the leading order of the range projection is
  \begin{equation}
    \label{e.Pexp-R}
    T k_\R + \eps^2 \Bigl[ \partial^2_\eps P(v^*,0,0) \Bigr]_\R = 0 \ .
  \end{equation}
  Thus have that $k_\R$ can be solved as a function of $k_\K$ and
  is of order $\eps^2$ plus corrections. Since $\Bigl[ T k_\R
    \Bigr]_\K = 0$, the leading order of the kernel equation of the
  Lyapunov-Schmidt decomposition is
  \begin{equation}
    \label{e.Pexp-K}
    \Bigl[ B(k_\R,k_\K) + \eps^2 \partial^2_\eps P(v^*,0,0) \Bigr]_\K = 0 \ ,
  \end{equation}
  where we omitted the terms $\eps L k_\R$ and $B(k_\R,k_\R)$
  respectively of order 3 and 4 in $\eps$. But in \eqref{e.Pexp-K},
  the term $B(k_\R,k_\K)$ is of order higher than 2, since $k_\K$ has
  to vanish with $\eps$: thus the necessary condition
  $\Bigl[\partial^2_\eps P(v^*,0,0) \Bigr]_\K = 0$, which is
  exactly~\eqref{necessary.IFT} once we recall that $\partial^2_\eps
  P(v^*,0,0)=\Pi_K L h^{(2,0)}(v^*;0,0)$.
\end{proof}

As we said before, for the totally degenerate cases, i.e. when
$dP(v^*;0,0)\equiv0$, we have to consider higher orders terms of
$P$. We report below the second order ones (dropping the dependence on
$(v;k,\eps)$ in the various $h^{(m,n)}$),

\begin{equation}
  \label{e.K.main}
  \begin{aligned}
    P_2(v;k,\eps) = \eps^2 \phantom{\Bigl(} & \Pi_K L h^{(2,0)} + \null
  \\ 
    \null + \eps  \Bigl( & \Pi_K L h^{(1,1)}
      - \frac34 \tond{2\Re\tond{v\overline{h}^{(2,0)}} + \bigl|{h^{(1,0)}}\bigr|^2}k
    \Bigr) + \null
  \\ 
    \null + \phantom{\eps} \Bigl( & \Pi_K Lh^{(0,2)}
      - \frac32 \Re\tond{v\overline{h}^{(1,1)}} k \Bigr) \ .
\end{aligned}
\end{equation}

\subsection{Existence and nonexistence}

Considering all the candidate $v^*$ satisfying~\eqref{e.bif.point}, we
analyze here their continuation to solution of the full kernel
equation, and thus of our original problem.  The first step is to
check if it is possible to apply Lemma~\ref{l.nondegen.IFT}.  It is
tedious yet straightforward to check that
\begin{equation}
\label{e.Pe=0}
  \Pi_K L h^{(1,0)}(v^*;0,0) = 0\ ,
\end{equation}
i.e. $\partial_\eps P(v^*,0,0)=0$.  For these as well as the
forthcoming calculations one could check also the Appendix. As a
consequence, condition~\eqref{nondegen.IFT} of
Lemma~\ref{l.nondegen.IFT} reduces to
\begin{equation}
  \label{reduced.nondegen.IFT}
  {\rm rk} (T) = 3
  \ .
\end{equation}

\paragraph{Existence of the continuation for $(0,0,0)$ and $(\pi,\pi,\pi)$.}
It is straightforward to verify that for the two isolated candidates,
i.e. those with phase shifts $(0,0,0)$ and $(\pi,\pi,\pi)$, the
matrices representing $T$ are respectively
\begin{equation*}
  \begin{pmatrix}
    -2 &  1 &  0 &  1 \\
     1 & -2 &  1 &  0 \\
     0 &  1 & -2 &  1 \\
     1 &  0 &  1 & -2 \\
  \end{pmatrix}
\quad\hbox{and}\qquad
  \begin{pmatrix}
     2 & -1 &  0 & -1 \\
    -1 &  2 & -1 &  0 \\
     0 & -1 &  2 & -1 \\
    -1 &  0 & -1 &  2 \\
  \end{pmatrix}\ ,
\end{equation*}
both with a single zero eigenvalue associated to the gauge direction
$e^{\im\theta}(1,1,1,1)$, and thus rank equal to 3, so that
Lemma~\ref{l.nondegen.IFT} applies.

\paragraph{Non existence of the continuation for the family
$(\varphi,\varphi,\pi-\varphi)$.}  Let us start with the first of the
three families of candidates, the asymmetric vortices, ignoring the
straightforward cases $\varphi\in\{0,\pi\}$, and
$\varphi=\pm\frac{\pi}2$ of the symmetric vortices left for a
subsequent analysis. The matrix representing $T$ is
\begin{equation*}
  T = \frac{\cos(\varphi)}2
  \begin{pmatrix}
 0            & e^{-\im\varphi} & 0            & e^{-\im\varphi}
\\
 e^{\im\varphi} & -2            & e^{-\im\varphi} & 0
\\
 0            & e^{\im\varphi}  & 0            & e^{\im\varphi}
\\
 e^{\im\varphi} & 0             & e^{-\im\varphi} & 2
  \end{pmatrix}\ .
\end{equation*}
Such a matrix has a double zero eigenvalue, associated as expected to
the gauge and the family directions, respectively
$e^{\im\theta}(1,e^{\im\varphi},e^{\im2\varphi},-e^{\im\varphi})$ and
$\partial_\varphi v^*(\varphi) = e^{\im(\theta+\varphi)}
(0,1,2e^{\im\varphi},-1)$. Therefore its rank is equal to 2, so that
Lemma~\ref{l.nondegen.IFT} does not apply. We therefore check for non
existence via Lemma~\ref{l.degen.IFT}: recalling~\eqref{e.Pe=0}, one
has to check only the necessary condition \eqref{necessary.IFT.2}, and
the calculations show that
\begin{equation*}
  \Pi_K L h^{(2,0)}= \frac{R\sin(\varphi)}{4\omega^2}\im e^{\im\theta}
  (1,0,0,e^{\im\varphi})\ ,
\end{equation*}
which is indeed non orthogonal to $\partial_\varphi v^*(\varphi)$.

\paragraph{Non existence of the continuation for the family
$(\varphi,\pi-\varphi,\pi-\varphi)$.}  As before, ignoring the special
cases $\varphi\in\{0,\pi\}$ and $\varphi=\pm\frac{\pi}2$, the matrix
representing $T$ is
\begin{equation*}
  T = \frac{\cos(\varphi)}2
  \begin{pmatrix}
-2            & e^{-\im\varphi} & 0            & e^{\im\varphi}
\\
 e^{\im\varphi}  & 0            & e^{\im\varphi} & 0
\\
 0            & e^{-\im\varphi}  & 2            & e^{\im\varphi}
\\
 e^{-\im\varphi} & 0             & e^{-\im\varphi} & 0
  \end{pmatrix}\ .
\end{equation*}
Such a matrix has a double zero eigenvalue, associated as expected
to the gauge and the family directions, respectively
$e^{\im\theta}(1,e^{\im\varphi},-1,e^{-\im\varphi})$ and
$\partial_\varphi v^*(\varphi) = \im e^{\im(\theta+\varphi)}
(0,1,0,-e^{-\im2\varphi})$. Therefore its rank is equal to 2, so that
Lemma~\ref{l.nondegen.IFT} does not apply. Looking for non existence
via Lemma~\ref{l.degen.IFT}, we again has to check only the necessary
condition \eqref{necessary.IFT.2}, and the calculations show that
\begin{equation*}
  \Pi_K L h^{(2,0)}= \frac{R\sin(\varphi)}{4\omega^2}\im e^{\im\theta}
  (1,0,0,-e^{-\im\varphi})\ ,
\end{equation*}
which is indeed non orthogonal to $\partial_\varphi v^*(\varphi)$.

\paragraph{Non existence of the continuation for the family
$(\varphi,\pi-\varphi,\varphi)$.}  As before, ignoring the special
cases $\varphi\in\{0,\pi\}$ and $\varphi=\pm\frac{\pi}2$, the
matrix representing $T$ is
\begin{equation*}
  T = \frac{\cos(\varphi)}2
  \begin{pmatrix}
 0            & e^{-\im\varphi} & 0            & e^{-\im\varphi}
\\
 e^{\im\varphi}  & 0            & e^{\im\varphi} & 0
\\
 0            & e^{-\im\varphi}  & 0            & e^{-\im\varphi}
\\
 e^{\im\varphi} & 0             & e^{\im\varphi} & 0
  \end{pmatrix}\ .
\end{equation*}
Such a matrix has a double zero eigenvalue, associated as expected
to the gauge and the family directions, respectively
$e^{\im\theta}(1,e^{\im\varphi},-1,-e^{\im\varphi})$ and
$\partial_\varphi v^*(\varphi) = \im e^{\im(\theta+\varphi)}
(0,1,0,-1)$. Therefore its rank is equal to 2, so that
Lemma~\ref{l.nondegen.IFT} does not apply. Looking for non-existence
via Lemma~\ref{l.degen.IFT}, we again has to check only the necessary
condition \eqref{necessary.IFT.2}, and the calculations show that
\begin{equation*}
  \Pi_K L h^{(2,0)}= \frac{R\sin(\varphi)}{4\omega^2}\im e^{\im\theta}
  (1,0,0,e^{\im\varphi})\ ,
\end{equation*}
which is indeed non orthogonal to $\partial_\varphi v^*(\varphi)$.

\paragraph{Non existence of the continuation for the vortices
$\pm\left(\frac{\pi}2,\frac{\pi}2,\frac{\pi}2\right)$.}  The two
vortices are located at the intersection of the above analyzed three
families, and due to this fact the degeneracy is complete. It is
thus not surprising that $T$ turns out to be the null matrix,
i.e. $dP(v^*;0,0)=0$.  The terms we have to investigate are those
given in the second order expansion of $P$, written
in~\eqref{e.K.main}.

After some calculations, we get that
\begin{align*}
  \Pi_K L h^{(2,0)} &= \frac{R}{4\omega^2}e^{\im\theta} (\im,0,0,-1)\ ,
\\  
  \Pi_K L h^{(1,1)} &= - \frac{k}{4\omega}\ ,
\\
  \Pi_K Lh^{(0,2)} &= \frac{3R}{4\omega}e^{\im\theta}
                    (k_1^2, \im k_2^2, -k_3^2, -\im k_4^2)\ ,
\end{align*}
and

\begin{align*}
  - \frac34 \tond{2\Re\tond{v\overline{h}^{(2,0)}} + \bigl|{h^{(1,0)}}\bigr|^2}k &=
    \frac{k}{4\omega}\ ,
\\    
  - \frac32 \Re\tond{v\overline{h}^{(1,1)}} k &=
       \frac{3R}{8\omega}e^{\im\theta} (\im k_1(k_2-k_4),  k_2(k_1-k_3),
                                      \im k_3(k_2-k_4), k_4(k_1-k_3))\ .
\end{align*}

In particular, the mixed derivatives $D_k\partial_\eps P(v^*;k,h)[k]$
vanish. Setting $P_2(v^*;k,h)=0$, one obtains the system
\begin{equation}
  \label{e.K.system}
  \left\{
  \begin{aligned}
    2k_1^2 + \im k_1(k_2-k_4)    &= -\frac{2\im}{3\omega}
\\
    2\im k_2^2 + k_2(k_1-k_3)    &= 0
\\
   -2k_3^2 + \im k_3(k_2-k_4)    &= 0
\\
   -2\im k_4^2 + k_4(k_1-k_3)    &= \frac2{3\omega}
  \end{aligned}
  \right.\ ,
\end{equation}
which is clearly impossible: indeed, keeping in mind that $k_j$ are
real variables, it is evident in the first equation that the
l.h.t. can be pure imaginary only when $k_1=0$, which however implies
$0=- 2\im/3\omega$. A similar argument can be developed on the last
equation, playing with $k_4$. Since for the vortices we have $P=P_2 +
\text{h.o.t.}$, the nonexistence of solutions for $P_2=0$ implies the
non existence of solutions for the full equation $P=0$.

\begin{remark}
The proof of Theorem~\ref{t.nonexistence} based on the
Lyapunov-Schmidt decomposition, assuming only the continuity of the
solution, is more general than the one obtained through the
perturbative approach, that requires a $C^2$ regularity.  Conversely,
from a practical point of view, the perturbative method allows a
direct approach that can be easily and efficiently implemented via
algebraic manipulation.  This is maybe even more evident in the
analogue problem on other dNLS models, like the zigzag one.
\end{remark}

\section{Discussion and Conclusions}

In the present work, we have revisited the topic of discrete solitons
and vortices in lattice models of the discrete nonlinear
Schr{\"o}dinger type.  Motivated by the interest in examining 2d
asymmetric and super-symmetric configurations, but also by the desire
to have a setting that is more analytically tractable, we came up with
a 1d toy model. The latter emulates a key feature of the 2d lattice
through the inclusion (in a 1d chain) of interactions with the
next-to-next-nearest neighbor.  To leading order, the persistence
conditions for a four-site configuration suggest the possibility of,
not only ``standard'' solutions with phase differences of $0,~\pi$,
but also of ones involving relative phases of $\pi/2$ (super-symmetric
vortex-like states), and of asymmetric ones involving a free parameter
$\vphi$, strongly reminiscent of those explored in~\cite{PelKF05b} at
the discrete level and in~\cite{tristram} at the level of continuum
problems with periodic potentials.

To address the question of persistence of such states, we have
utilized two different methods; one involved a conserved quantity upon
the assumption of $\C^2$ regularity of the solutions with respect to
the small coupling $\eps$, while the second one required only
continuity, using a Lyapunov-Schmidt reduction (and projections to the
resulting kernel and range equations) to derive necessary and
sufficient conditions for the persistence of the different solution
families.  The surprising finding is that among all the possible
solutions {\it only} the ones with trivial phase shifts of $0$ and
$\pi$ can be found to persist in a four-sites configuration.

This result raises some intriguing questions.  A natural one is how
general is this result. Admittedly, in the case of every model, the
relevant conserved quantity or LS reduction need to be re-performed
and the answer has to be given on a case by case basis.  However, our
experience so far with 1d chains suggests that for generalized dNLS
models, such a conservation law should be traceable more generally and
the conclusions are likely to be similar for super-symmetric and
asymmetric families of vortices. On the other hand, we realize that
this topic requires considerable additional examination for a
conclusive closure. At the same time, understanding the similarities
and differences of the present setting with that of the 2d problem and
the supersymmetric vortices therein (as well as how these
considerations extend to higher dimensional settings) is a
particularly intriguing question.  Such topics are currently under
consideration and will be reported in future publications.

\section*{Acknowledgments:}
We thank Prof. Dmitry Pelinovsky for the helpful discussions on the
proof of Lemma 4.4 that we had during the SIAM conference on Nonlinear
Waves and Coherent Structures. This work has been partially supported
by NSF-PHY-1602994 (PGK). Also, VK and PGK 
acknowledge that this work made possible by NPRP grant {\#} [9-329-1-067] 
from Qatar National Research Fund (a member of Qatar Foundation).


\def\cprime{$'$} \def\i{\ii}\def\cprime{$'$} \def\cprime{$'$}

\end{document}